\documentclass[12pt,reqno]{amsart}
\usepackage{amsfonts}
\usepackage{graphicx}
\usepackage{amscd}
\usepackage{amsmath}
\usepackage{epsfig}
\usepackage[usenames]{color}
\usepackage{subfigure}
\usepackage[english]{babel}

\providecommand{\U}[1]{\protect\rule{.1in}{.1in}}
\providecommand{\U}[1]{\protect\rule{.1in}{.1in}}
\allowdisplaybreaks
\newtheorem{theorem}{Theorem}
\newtheorem{lemma}{Lemma}
\newtheorem{proposition}{Proposition}
\theoremstyle{plain}
\newtheorem{acknowledgement}{Acknowledgement}

\newtheorem{example}{Example}

\newtheorem{remark}{Remark}

\numberwithin{equation}{section}

\DeclareMathOperator{\sech}{sech}

\DeclareMathOperator{\csch}{csch}

\begin{document}
\title[Nonlinear Schr$\ddot{\mbox{O}}$dinger Equation]{Blow-up results and
soliton solutions for a generalized variable coefficient nonlinear Schr\"{o}%
dinger equation}
\author{J. Escorcia}
\address{Department of Mathematics, University of Puerto Rico, Arecibo, P.O.
Box 4010, Puerto Rico 00614-4010.}
\email{jose.escorcia@upr.edu}
\author{E. Suazo}
\address{School of Mathematical and Statistical Sciences, University of
Texas Rio Grande Valley, 1201 W. University Drive, Edinburg, Texas,
78539-2999.}
\email{erwin.suazo@utrgv.edu}
\date{\today }

\begin{abstract}
In this paper, by means of similarity transformations we study exact
analytical solutions for a generalized nonlinear Schr$\ddot{\mbox{o}}$dinger
equation with variable coefficients. This equation appears in literature
describing the evolution of coherent light in a nonlinear Kerr medium,
Bose-Einstein condensates phenomena and high intensity pulse propagation in
optical fibers. By restricting the coefficients to satisfy Ermakov-Riccati
systems with multiparameter solutions, we present conditions for existence
of explicit solutions with singularities and a family of oscillating
periodic soliton-type solutions. Also, we show the existence of bright-,
dark- and Peregrine-type soliton solutions, and by means of a computer
algebra system we exemplify the nontrivial dynamics of the solitary wave
center of these solutions produced by our multiparameter approach.


\end{abstract}

\maketitle

\section{ Introduction}

The study of the nonlinear Schr\"{o}dinger equation (NLS) with real
potential $V$ 
\begin{equation}
i\psi _{t}=-\frac{1}{2}\Delta \psi +V(\boldsymbol{x,}t)\psi +\lambda
\left\vert \psi \right\vert ^{2s}\psi ,\qquad \psi (0,\boldsymbol{x}%
)=\varphi (\boldsymbol{x}),\qquad \mathbf{\boldsymbol{x}}\in \mathbb{R}%
^{n},\qquad \Delta =\sum_{j=1}^{n}\partial _{x_{j}x_{j}}  \label{NLSV}
\end{equation}%
has been studied extensively not only for its role in physics, such as in
Bose-Einstein condensates and nonlinear optics, but also for its
mathematical complexity (for a review of the several results available see 
\cite{Abl}, \cite{Agr}, \cite{Caz:Ha}, \cite{Caz}, \cite{Carretero}, \cite%
{Li:Ve} and \cite{Sulem:Sulem}). For the case $\lambda =-1$, $V\equiv 0,$
and $ns<2$ (subcritical case) Weinstein \cite{Weinstein} proved that if $%
\varphi \in H^{1},$ then $\psi $ exists globally in $H^{1}.$ It is also
known (see \cite{Caz}, \cite{FibichBook} and \cite{Sulem:Sulem} for a
complete review) that NLS for critical ($ns=2$) and supercritical ($ns>2$)
cases present solutions that become singular in a finite time in $L^{p}$ for
some finite $p$. In \cite{Fibich2011} singular solutions of the subcritical
NLS were presented in $L^{p}$.

In \cite{Carles2011} it was proved that if $\varphi \in \Sigma =\{f\in H^{1}(%
\mathbb{R}^{n}):\boldsymbol{x}\rightarrow \left\vert \boldsymbol{x}%
\right\vert f(\boldsymbol{x})\in L^{2}(\mathbb{R}^{n})\}$, $V(\boldsymbol{x,}%
t)$ is real, locally bounded in time and subquadratic in space, and $\lambda
\in \mathbb{R}$, then the solution of the Cauchy initial value problem
exists globally in $\Sigma ,$ provided that $s<2/n$ or $s\geq 2/n$ and $%
\lambda \geq 0.$ Also, in \cite{Carles2011} it was shown that if $V(%
\boldsymbol{x,}t)=b(t)x_{j}^{2},$ $b(t)\in C(\mathbb{R};\mathbb{R})$ in (\ref%
{NLSV}), then there exist blow-up solutions if $\lambda <0$ and $s=2/n.$ The
proof uses the generalized Melher's formula introduced in \cite%
{Cor-Sot:Lop:Sua:Sus}. In \cite{Perez:Torres:Konotop} and \cite{Sus} a
generalized pseudoconformal transformation (lens transform in optics \cite%
{Talanov}) was presented. In this paper, as a first main result we will use
a generalized lens transformation to construct solutions with finite-time
blow-up in $L^{p}$\ norm for $1\leq p\leq \infty $ of the general variable
coefficient nonlinear Schr\"{o}dinger: 
\begin{equation}
i\psi _{t}=-a\left( t\right) \psi _{xx}+(b\left( t\right) x^{2}-f\left(
t\right) x+G(t))\psi -ic\left( t\right) x\psi _{x}-id\left( t\right) \psi
+ig\left( t\right) \psi _{x}+h\left( t\right) \left\vert \psi \right\vert
^{2s}\psi .  \label{NLSVC}
\end{equation}

In modern nonlinear sciences some of the most important models are the
variable coefficient nonlinear Schr\"{o}dinger-type ones. Applications
include long distance optical communications, optical fibers and plasma
physics,\ see \cite{Agr}, \cite{Al}, \cite{Bru}, \cite{Chen}, \cite{Dai:Wang}%
, \cite{He}, \cite{He:Li}, \cite{He:Chara:Ke:Fran}, \cite{Krug:Pea:Har}, 
\cite{Ma:Shabat}, \cite{Pono:Agrawal}, \cite{Pono:Agrawal2}, \cite%
{Ra:Agrawal}, \cite{Serkin:Hasegawa}, \cite{Serkin:Matsumoto}, \cite{Tian}, 
\cite{Wa:Tian}, \cite{Yu:Yan} and references therein.

If we make $a(t)=\Lambda /4\pi n_{0},$ $\Lambda $ being the wavelength of
the optical source generating the beam, and choose $c(t)=g(t)=0,$ then (\ref%
{NLSVC}) models a beam propagation inside of a planar graded-index nonlinear
waveguide amplifier with quadratic refractive index represented by $b\left(
t\right) x^{2}-f\left( t\right) x+G(t),$ and $h\left( t\right) $ represents
a Kerr-type nonlinearity of the waveguide amplifier, while $d\left( t\right) 
$ represents the gain coefficient. If $b\left( t\right) >0$ \cite%
{Pono:Agrawal} $($resp. $b\left( t\right) <0,$ see \cite{Ra:Agrawal}$)$ in
the low-intensity limit, the graded-index waveguide acts as a linear
defocusing (focusing) lens.

Depending on the selections of the coefficients in equation (\ref{NLSVC}),
the applications vary in very specific problems (see \cite{Tian} and
references therein):

\begin{itemize}
\item Bose-Einstein condensates \cite{Carretero}: $b(\cdot )\neq 0$, $a$, $h$
constants and other coefficients are zero$.$

\item Dispersion-managed optical fibers and soliton lasers \cite%
{Krug:Pea:Har}, \cite{Serkin:Hasegawa} and \cite{Serkin:Matsumoto}: $a(\cdot
),$ $h(\cdot ),$ $d(\cdot )\neq 0$ \ are respectively dispersion,
nonlinearity and amplification, and\ the other coefficients are zero. $%
a(\cdot )$ and $h(\cdot )$ can be periodic as well, see \cite%
{Ablowitz:Hirooka} and \cite{Medvedev:Shty:Mush}.

\item Pulse dynamics in the dispersion-managed fibers \cite{Ma:Shabat}: $%
h(\cdot )\neq 0,$ $a$ is a constant and other coefficients are zero$.$
\end{itemize}

In this paper to obtain the main results we use a fundamental approach
consisting of the use of similarity transformations and the solutions of
Riccati Ermakov systems with several parameters inspired by the work in \cite%
{Marhic78}. Similarity trasformations have been a very popular strategy in
nonlinear optics since the lens transform presented by Talanov \cite{Talanov}%
; extensions of this approach have been presented in \cite%
{Perez:Torres:Konotop} and \cite{Sus}. Applications include nonlinear
optics, Bose-Einstein condensates, integrability of NLS and quantum
mechanics, see for example \cite{Al}, \cite{Aldaya:Guerrero}, \cite%
{Carles2011}, \cite{Lo:Su:VeSy} and references therein. E. Marhic in 1978
introduced (probably for the first time) a one-parameter $\left\{ \alpha
(0)\right\} $ family of solutions for the linear Schr\"{o}dinger equation of
the one-dimensional harmonic oscillator; the use of an explicit formulation
(classical Melher's formula \cite{Fey:Hib} and \cite{Ni:Su:Uv}) for the
propagator was fundamental. The solutions presented by E. Marhic constituted
a generalization of the original Schr\"{o}dinger wave packet with
oscillating width. Also, in \cite{Cor-Sot:Lop:Sua:Sus} a generalized
Melher's formula for a general linear Schr\"{o}dinger equation of the
one-dimensional generalized harmonic oscillator of the form (\ref{NLSVC})
with $h(t)=0$ was presented$.$ For the latter case in \cite{Lan:Lop:Sus}, 
\cite{Lop:Sus:VegaGroup} and \cite{Suaz:Suspre}, multiparameter solutions in
the spirit of Marhic in \cite{Marhic78} have been presented. The parameters
for the Riccati system arose originally in the process of proving
convergence to the initial data for the Cauchy initial value problem (\ref%
{NLSVC}) with $h(t)=0$ and in the process of finding a general solution of a
Riccati system \cite{Sua1} and \cite{Suaz:Suspre}. Ermakov systems with
solutions containing parameters \cite{Lan:Lop:Sus} have been used
successfully to construct solutions for the generalized harmonic oscillator
with a hidden symmetry \cite{Lop:Sus:VegaGroup}, and they have also been
used to present Galilei transformation, pseudoconformal transformation and
others in a unified manner, see \cite{Lop:Sus:VegaGroup}. More recently they
have been used in \cite{Mah:Su:Sus} to show spiral and breathing solutions
and solutions with bending for the paraxial wave equation. In this paper, as
a second main result we introduce a family of Schr\"{o}dinger equations
presenting periodic soliton solutions by using multiparameter solutions for
Riccati-Ermakov systems. Further, as a third main result we show that these
parameters provide a control on the dynamics of solutions for equations of
the form (\ref{NLSVC}).\ These results should deserve numerical and
experimental studies.

This paper is organized as follows: In Section 2, as an application of a
generalized lens transformation and multiparameter solutions for Riccati
systems we present conditions to obtain solutions with singularity in finite
time in $L^{p}$\ norm, $1\leq p\leq \infty $ for (\ref{NLSVC}). Also, we
show that through this more general parameter approach we can obtain the
same $L^{\infty }$ solutions with finite-time blow-up for standard NLS
presented in \cite{Cor-Sot:Lop:Sua:Sus} and finite-time blow-up for NLS with
quadratic potential. In Section 3, we present a family of soliton solutions
for (\ref{NLSVC}) presenting bright- and dark-type solitons; this family
includes the standard NLS models. This family has multiparameter solutions
coming from solutions of a related Ermakov system, extending the results
presented in \cite{Sua:Sus}, where a Riccati system was used. By the use of
these parameters the dynamics of periodic solutions for (\ref{NLSVC}) show
bending properties, see Figures 1 and 2. In Section 4, again, as an
application of generalized lens transformations and an alternative approach
to solve the Riccati system (\ref{rica1})-(\ref{rica6}) we present how the
parameters provide us with a control on the center axis of the solution of
bright and dark soliton solutions for special coefficients in (\ref{NLSVC}).
Figures 3 and 4 show the bending propagation of the solutions after
introducing parameters, extending the results presented in \cite{Mah:Su:Sus}
and \cite{Sua:Sus} to (\ref{NLSVC}). Also we show that it is possible to
construct a transformation that reduces (\ref{NLSVC}), with $a(t)=$ $%
l_{0}=\pm 1$ and $G(t)=0,$ to standard NLS with convenient initial data
(Lemma 4) in order to assure existence and uniqueness of classical solutions
(Proposition 1). As an application we show how the dynamics of the Peregrine
soliton solutions of the nonlinear Schr\"{o}dinger equation consider change
when the dissipation, $d(t)$, and the nonlinear term, $h(t)$ change, see
Figures 5-8. We have also prepared a Mathematica file as supplemental
material where all the solutions for this Section are verified. Finally, in
Section 5 we have an appendix recalling the main tools we have used for our
results. These tools are a solution with multiparameters of the Riccati
system (\ref{rica1})-(\ref{rica6}) and a modification of the transformation
introduced in \cite{Sus}; we have introduced an extra parameter $l_{0}=\pm 1$
in order to use standard solutions for Peregrine-type soliton solutions.
Also a 2D version of a generalized lens transformation is recalled. All the
formulas from the appendix have been verified previously using computer
algebra systems \cite{Ko:Su;su}.

\section{Finite-time blow-up for nonautonomous nonlinear Schr\"{o}dinger
equations}

In this section as an application of the multiparameter solution for Riccati
systems we present conditions needed in order to obtain solutions with
singularities in finite time with $L^{\infty }$ norm for (\ref{NLSVC}). We
show that we can obtain the same $L^{\infty }$ solutions with finite-time
blow-up for standard NLS presented in \cite{Cor-Sot:Lop:Sua:Sus} and
finite-time blow-up for the Gross-Pitaevskii equation. Also, as an
application of a generalized lens transformation (see section 6.2) in this
section we present conditions to obtain solutions with singularity in finite
time with $L^{p}$ norm for (\ref{NLSVC}) for dimensions one and two. We
present our first main result:

\begin{theorem}[Solutions with singularity in finite time with $L^{p}$ norm, 
$1\leq p\leq \infty $]
If the characteristic equation associated to (\ref{NLSVC}), i.e 
\begin{equation}
\mu ^{\prime \prime }-\left( \frac{a^{\prime }}{a}-2c+4d\right) \mu ^{\prime
}+4\left( ab-cd+d^{2}+\frac{d}{2}\left( \frac{a^{\prime }}{a}-\frac{%
d^{\prime }}{d}\right) \right) \mu =0,
\end{equation}%
admits two standard solutions $\mu _{0}$ and $\mu _{1}$ subject to%
\begin{equation}
\mu _{0}\left( 0\right) =0,\quad \mu _{0}^{\prime }\left( 0\right) =2a\left(
0\right) \neq 0\qquad \mu _{1}\left( 0\right) \neq 0,\quad \mu _{1}^{\prime
}\left( 0\right) =0,  \label{initicond}
\end{equation}%
and if we choose $h(t)=a(t)\beta ^{2}(t)\mu ^{2s}(t)$, $\beta $ and $\mu $
satisfies a solvable Riccati-type system (\ref{nlse3})-(\ref{nlse8}) (with $%
\mu \left( 0\right) ,$ $\beta (0)\neq 0$), then there exists an interval $I$
of time such that if $-\alpha \left( 0\right) \in \gamma _{0}(I),$ then (\ref%
{NLSVC}) presents a solution with finite-time blow-up in $L^{p}$ norm at $%
T^{\ast }=\gamma _{0}^{-1}(-\alpha \left( 0\right) )\in I$. Further,
solutions present the following explicit form:

(i). If $p=\infty ,$ a solution for (\ref{NLSVC}) is given explicitly by 
\begin{equation}
\psi _{y}\left( x,t\right) =\frac{1}{\sqrt{\mu \left( t\right) }}\
e^{iS_{y}\left( x,t\right) },  \label{Anzats1}
\end{equation}%
where $y$ is a parameter and $S_{y}\left( t,x\right) =\alpha \left( t\right)
x^{2}+\beta \left( t\right) xy+\gamma \left( t\right) y^{2}+\delta \left(
t\right) x+\varepsilon \left( t\right) y+\kappa \left( t\right) $ and $%
\alpha (t),$ $\beta (t),$ $\gamma (t),$ $\delta (t),$ $\epsilon (t)$ and $%
\kappa (t)$ satisfy the Riccati system (\ref{nlse3})-(\ref{nlse8}).

(ii). If $1\leq p<\infty $, a solution for (\ref{NLSVC}) is given explicitly
by 
\begin{equation}
\psi (x,t)=\frac{1}{\sqrt{\mu (t)}}e^{i(\alpha (t)x^{2}+\delta (t)x+\kappa
(t))}\chi (\xi ,\tau ),\qquad \xi =\beta (t)x+\varepsilon (t),\qquad \tau
=\gamma (t),  \label{Anzats2}
\end{equation}%
where $\alpha (t),$ $\beta (t),$ $\gamma (t),$ $\delta (t),$ $\epsilon (t)$
and $\kappa (t)$ are given by (\ref{mu})-(\ref{kappa0})$.$ Additionally, $%
\chi $ satisfies%
\begin{equation}
i\chi _{\tau }+\chi _{\xi \xi }+\left\vert \chi \right\vert ^{2s}\chi =0.
\end{equation}

(iii). (2D case) The natural 2D version of (\ref{NLSVC}), the nonlinear
equation%
\begin{align}
i\psi _{t}& =-a\left( \psi _{xx}+\psi _{yy}\right) +b\left(
x^{2}+y^{2}\right) \psi -ic\left( x\psi _{x}+y\psi _{y}\right) -2id\psi
\label{2dNLS} \\
& -\left( xf_{1}+yf_{2}\right) \psi +i\left( g_{1}\psi _{x}+g_{2}\psi
_{y}\right) -\left\vert \psi \right\vert ^{2s}\psi ,  \notag
\end{align}%
where $a,$ $b,$ $c,$ $d,$ $f_{1,2}$ and $g_{1,2}$ are real-valued functions
of $t,$ admits an explicit solution with finite-time blow-up of the form 
\begin{equation}
\psi =\mu ^{-1}e^{i(\alpha (x^{2}+y^{2})+(\delta _{1}x+\delta _{2}y)+\kappa
_{1}+\kappa _{2})}\chi (\xi ,\eta ,\tau ),  \label{Anzats3}
\end{equation}%
where $\xi =\beta (t)x+\varepsilon _{1}(t),$ $\eta =\beta (t)y+\varepsilon
_{2}(t),$ $\tau =\gamma (t),$ $h(t)=a(t)\beta ^{2}(t)\mu ^{2s}(t)$, and $%
\alpha (t),$ $\beta (t),\gamma (t),$ $\delta _{1}(t),$ $\delta _{2}(t),$ $%
\kappa _{1}(t),$ $\kappa _{2}(t),$ $\varepsilon _{1}(t),$ $\varepsilon
_{2}(t)$ satisfy the given conditions in Lemma 4. Finally, $\chi $ is a
solution of 
\begin{equation}
i\chi _{\tau }+\chi _{\xi \xi }+\chi _{\eta \eta }+\left\vert \chi
\right\vert ^{2s}\chi =0.
\end{equation}
\end{theorem}

\begin{proof}
To prove (i) we follow \cite{Cor-Sot:Lop:Sua:Sus} and look for a solution of
the form (\ref{Anzats1}). After substituting on (\ref{NLSVC}) we obtain the
following Riccati system: 
\begin{align}
& \frac{d\alpha }{dt}+b\left( t\right) +2c\left( t\right) \alpha +4a\left(
t\right) \alpha ^{2}=0,  \label{nlse3} \\
& \frac{d\beta }{dt}+\left( c\left( t\right) +4a\left( t\right) \alpha
\left( t\right) \right) \beta =0,  \label{nlse4} \\
& \frac{d\gamma }{dt}+a\left( t\right) \beta ^{2}\left( t\right) =0,
\label{nlse5} \\
& \frac{d\delta }{dt}+\left( c\left( t\right) +4a\left( t\right) \alpha
\left( t\right) \right) \delta =f\left( t\right) +2\alpha \left( t\right)
g\left( t\right) ,  \label{nlse6} \\
& \frac{d\varepsilon }{dt}=\left( g\left( t\right) -2a\left( t\right) \delta
\left( t\right) \right) \beta \left( t\right) ,  \label{nlse7} \\
& \frac{d\kappa }{dt}=g\left( t\right) \delta \left( t\right) -a\left(
t\right) \delta ^{2}\left( t\right) -\frac{h\left( t\right) }{\mu ^{s}\left(
t\right) }.  \label{nlse8}
\end{align}%
Using (\ref{mu})-(\ref{kappa0}) in the appendix, (\ref{nlse3})-(\ref{nlse7})
can be solved, but (\ref{nlse8}) absorbs the nonlinearity and must be solved
separately. Since there exists an interval $J$ of time with $\mu _{0}\left(
t\right) \neq 0$ for all $t\in J,$ and $\mu _{0}\left( t\right) $ and $\mu
_{1}\left( t\right) $ have been chosen to be linearly independent on an
interval, let's say $J^{\prime },$ we observe that for $t\in J\cap J^{\prime
}\equiv I,$ we get%
\begin{equation*}
\gamma _{0}^{\prime }(t)=\frac{W[\mu _{0}\left( t\right) ,\mu _{1}\left(
t\right) ]}{2\mu _{0}^{2}(t)}\neq 0,
\end{equation*}%
and therefore from the general expression for $\mu $ given by (\ref{mu}),
see \cite{Co:Suslov},the equation (\ref{Anzats1}) will have finite-time
blow-up at $T^{\ast }=\gamma _{0}^{-1}(-\alpha \left( 0\right) )\in I.$

Now we proceed to prove (ii). Using the generalized lens transform, see
Lemma 2, we can transform the nonautonomous and inhomogeneous nonlinear Schr%
\"{o}dinger equation (\ref{NLSVC}) into the standard one$:$%
\begin{equation}
i\chi _{\tau }+\chi _{\xi \xi }+\left\vert \chi \right\vert ^{2s}\chi =0,
\label{autoNLS}
\end{equation}%
recalling (\cite{FibichBook}, \cite{Li:Ve} and \cite{Sulem:Sulem}) that the
autonomous focusing NLS (\ref{autoNLS}) in dimension $n$ allows solutions of
the form $\chi $ $=e^{i\tau }R(r),$ where $r=\left\vert \xi \right\vert $
and $R$ is the solution of 
\begin{equation}
R^{\prime \prime }(r)+\frac{n-1}{r}R^{\prime }(r)-R(r)+R^{2s+1}(r)=0,\qquad
R^{\prime }(0)=0,\qquad R(\infty )=0.
\end{equation}%
In particular, for \ $n=1,$ a solution of (\ref{autoNLS}) is given by 
\begin{equation}
R(\xi )=(s+1)^{1/2s}\cosh ^{-1/s}\left( s\xi \right) ,
\end{equation}%
and for all $t\in R$ and $p\in \lbrack 1,\infty ]$ we have $\left\vert
\left\vert R\right\vert \right\vert _{p}<\infty .$ Therefore, we obtain a
solution with $L^{p}$ finite-time blow-up for (\ref{NLSVC}) in time of the
form%
\begin{equation}
\psi (x,t)=\frac{1}{\sqrt{\mu (t)}}e^{i(\alpha (t)x^{2}+\delta (t)x+\kappa
(t))}\chi (\xi ,\tau )=\frac{1}{\sqrt{\mu (t)}}e^{i(\alpha (t)x^{2}+\delta
(t)x+\kappa (t))}e^{i\tau }R(\xi ).
\end{equation}%
Again, since $\mu (t)=2\mu \left( 0\right) \mu _{0}\left( t\right) \left(
\alpha \left( 0\right) +\gamma _{0}\left( t\right) \right) $ we can predict
a blow-up at $T^{\ast }=\gamma _{0}^{-1}(-\alpha (0))$.

To prove (iii), we consider the unique positive radial solution to%
\begin{equation}
\Delta Q(\rho )-Q(\rho )+\left\vert Q(\rho )\right\vert ^{1+4/n}Q(\rho )=0,
\label{Elliptic2D}
\end{equation}

usually referred to as the ground state$.$ $Q$ vanishes at infinity (see 
\cite{Kwong} and \cite{Tao}). Similar to the one dimensional case, Lemma 3
provides a blow-up solution for (\ref{NLSVC}) with $2s=$ $1+4/n$ given by 
\begin{equation*}
\psi (\rho ,t)=\mu ^{-1}(t)e^{i(\alpha (t)(x^{2}+y^{2})+(\delta
_{1}(t)x+\delta _{2}(t)y)+\kappa _{1}(t)+\kappa _{2}(t))}Q(\rho ).
\end{equation*}%
This provides an example of an explicit blow-up solution $\left\vert
\left\vert \psi (t)\right\vert \right\vert _{p}$ $\rightarrow \infty $ as $%
t\rightarrow T^{\ast }=\gamma _{0}^{-1}(-\alpha (0))$ for the nonautonomous
nonlinear Schr\"{o}dinger equation (\ref{NLSVC}) in the two-dimensional case.
\end{proof}

The Theorem 1 above allows us to predict in an independent way the $%
L^{\infty }$ solution with finite-time blow-up for standard NLS found in 
\cite{Cor-Sot:Lop:Sua:Sus} in 2008.

\begin{example}
If we consider the equation 
\begin{equation}
i\frac{\partial \psi }{\partial t}=-\frac{1}{2}\frac{\partial ^{2}\psi }{%
\partial x^{2}}+h\left\vert \psi \right\vert ^{2s}\psi ,\qquad h=\text{%
constant},\quad s\geq 0,  \label{nlse11}
\end{equation}%
we can construct a solution with finite-time blow-up for the corresponding
Riccati system (\ref{nlse3})-(\ref{nlse8}). When we look for a solution of
the form (\ref{Anzats1}) it is easy to see that $\alpha _{0}(t)=\gamma
_{0}(t)=1/2t$ and $\beta _{0}(t)=-1/t,$ and by (\ref{mu})-(\ref{kappa}) the
solution of the corresponding characteristic equation is given by $\mu
\left( t\right) =2\mu (0)\alpha (0)t+\mu (0).$ Further we obtain explicitly 
\begin{eqnarray}
&&\alpha \left( t\right) =\frac{\alpha (0)}{1+2\alpha (0)t},\quad \beta
\left( t\right) =\frac{\beta (0)}{1+2\alpha (0)t},\quad \delta \left(
t\right) =\frac{\delta (0)}{1+2\alpha (0)t},  \label{nlse12} \\
&&\gamma \left( t\right) =\gamma (0)-\frac{\beta ^{2}(0)t}{2\left( 1+2\alpha
(0)t\right) },\qquad \qquad \varepsilon \left( t\right) =\varepsilon (0)-%
\frac{\beta (0)\delta (0)t}{1+2\alpha (0)t}.  \label{nlse13}
\end{eqnarray}

The equation (\ref{nlse8}) must be solved separately, and $\kappa (t)$ is
given by 
\begin{equation*}
\kappa \left( t\right) =\kappa (0)-\frac{\delta ^{2}(0)t}{2\left( 1+2\alpha
(0)t\right) }-\frac{h}{\alpha (0)}\xi _{s}\left( t\right)
\end{equation*}

with 
\begin{equation}
\xi _{s}\left( t\right) =\left\{ 
\begin{array}{ll}
\dfrac{1}{\left( 1-s\right) }\left( \left( \frac{1}{2}+t\alpha (0)\right)
^{1-s}-\left( \frac{1}{2}\right) ^{1-s}\right) , & \text{when }s\neq
1,\bigskip \\ 
\ln \left( 1+2t\alpha (0)\right) , & \text{when }s=1.%
\end{array}%
\right.  \label{nlse15}
\end{equation}%
Now, choosing $\alpha (0)=0,$ $\left\vert \psi \left( t,x\right) \right\vert
=1/\sqrt{2}$ is bounded at all times. However, when $\alpha (0)\neq 0,$ one
obtains 
\begin{equation}
\left\vert \psi \left( x,t\right) \right\vert =\frac{1}{\sqrt{\frac{1}{2}%
+t\alpha (0)}},\qquad t\geq 0,  \label{nlse15b}
\end{equation}%
which is bounded if $\alpha (0)>0,$ and blows up at a finite time $T^{\ast
}=-1/2\alpha (0)$ if $\alpha (0)<0.$ As expected, this result agrees with
the prediction of the theorem above, since $\gamma _{0}^{-1}(-\alpha (0))=$ $%
-1/2\alpha (0).$
\end{example}

The following example shows blow-up for the Gross-Pitaevskii equation:

\begin{example}
Let's consider the Gross-Pitaevskii equation%
\begin{equation}
i\psi _{t}=-(\psi _{xx}+\psi _{yy})+\sum_{j=1}^{2}\frac{\Omega (t)}{2}%
x_{j}^{2}\psi +\lambda |\psi |^{2}\psi .  \label{GP1}
\end{equation}

The characteristic equation associated to Gross-Pitaevskii equation is given
by%
\begin{equation}
\mu ^{\prime \prime }+\Omega (t)\mu =0.  \label{charaomega}
\end{equation}

Assuming $\Omega (t)$ is such that (\ref{charaomega}) allows two independent
solutions $\mu _{0}(t)$ and $\mu _{1}(t)$\ satisfying (\ref{initicond}),
then 
\begin{equation}
\alpha _{0}(t)=\frac{\mu _{0}^{\prime }(t)}{2\mu _{0}(t)},\qquad \beta
_{0}(t)=-\frac{1}{\mu _{0}(t)},\qquad \gamma _{0}(t)=\frac{\mu _{1}(t)}{2\mu
_{1}(0)\mu _{0}(t)}
\end{equation}

and \ 
\begin{equation}
\delta _{0}(t)=\varepsilon _{0}\left( t\right) =\kappa _{0}(t)=0.
\end{equation}

By Theorem 1, and using (\ref{mu})-(\ref{kappa0}) from the appendix, a
solution for (\ref{autoNLS}) is given by 
\begin{equation*}
\psi (\rho ,t)=\mu ^{-1}(t)e^{i(\alpha (t)(x^{2}+y^{2})+(\delta
_{1}(t)x+\delta _{2}(t)y)+\kappa _{1}(t)+\kappa _{2}(t))}Q(\rho ).
\end{equation*}

with $\mu \left( t\right) =2\mu \left( 0\right) \mu _{0}\left( t\right)
\left( \alpha \left( 0\right) +\gamma _{0}\left( t\right) \right) $

Considering $\alpha (0)=0,$ $\beta (0)=1,$ $\mu _{1}(0)=1$ and $\gamma
(0)=0, $ then 
\begin{equation}
\tau =\gamma (t)=\frac{\mu _{0}^{\prime }(t)}{2\mu _{1}(t)},\qquad \mu
(t)=\mu (0)\mu _{1}(t),\qquad \beta (t)=\frac{-\beta (0)}{\mu _{1}(t)}%
,\qquad \alpha \left( t\right) =\frac{\mu _{1}^{\prime }(t)}{2\mu _{1}(t)},
\end{equation}

and considering $\kappa _{i}(0)=0$ and for $i=1$ and $2$ we obtain 
\begin{equation*}
\delta _{i}(t)=\frac{\delta _{i}^{{}}(0)}{\mu _{1}(t)},\qquad \varepsilon
_{i}\left( t\right) =-\frac{\beta (0)\delta _{i}(0)\mu _{0}(t)}{\mu _{1}(t)}%
,\qquad \kappa _{i}(t)=-\frac{\left( \delta _{i}(0)\right) ^{2}\mu _{0}(t)}{%
2\mu _{1}(t)}.
\end{equation*}
\end{example}

\begin{remark}
Alternatively, we can use in this example the soliton solution $%
u(x,t)=e^{it}Q(x)$ to (\ref{autoNLS}) so that after applying the
pseudoconformal transform we can obtain solutions which blow up in finite
time (see a nice discussion in \cite{Tao}). Therefore, the following is a\
solution with blow up for (\ref{autoNLS}) given by%
\begin{equation}
\chi (x,t)=\frac{1}{t^{d/2}}Q\left( \frac{x}{t}\right) e^{i\left( \frac{%
\left\vert x\right\vert ^{2}}{2t}\right) -\frac{i}{t}}
\end{equation}%
where $Q$ is a ground state solution of (\ref{Elliptic2D}). Further, the
following is a solution for (\ref{GP1}) given by%
\begin{equation}
\psi (x,t)=\frac{e^{i(\alpha (t)x^{2}+\delta (t)x+\kappa (t))}}{\mu (t)\tau }%
Q\left( \frac{\boldsymbol{\xi }}{\tau }\right) e^{i\frac{\left\vert 
\boldsymbol{\xi }\right\vert ^{2}}{2\tau }-\frac{i}{\tau }},
\label{explicsol2d}
\end{equation}%
where $\boldsymbol{\xi }=(\beta (t)\boldsymbol{x+}\boldsymbol{\varepsilon }%
\left( t\right) )/\tau $, $\boldsymbol{x}=(x,y)$ and $\boldsymbol{%
\varepsilon }\left( t\right) =(\varepsilon _{1}\left( t\right) ,\varepsilon
_{2}\left( t\right) ),$ making $\delta _{i}^{{}}(0)=0,$ and after
simplification the module of (\ref{explicsol2d}) is 
\begin{equation}
\left\vert \psi (x,t)\right\vert =\left\vert \frac{1}{\mu (0)\mu _{0}(t)}%
Q\left( \frac{-\beta (0)\boldsymbol{x}}{\mu _{0}(t)}\right) \right\vert
\end{equation}

and then%
\begin{equation}
\lim_{t\rightarrow 0}\left\vert \left\vert \psi (t)\right\vert \right\vert
_{p}\rightarrow \infty .
\end{equation}
\end{remark}

It is possible to predict finite-time blow-up for toy examples:

\begin{example}
If we consider%
\begin{equation}
i\psi _{t}=-\psi _{xx}+\frac{x^{2}}{4}[\sin ^{2}t-\cos t]\psi +ix\sin t\psi
_{x}-i\sin t\psi -3e^{(3-3\cos t)}|\psi |^{2}\psi ,
\end{equation}%
the corresponding characteristic equation is given by $\mu ^{^{\prime \prime
}}-6\sin t\mu ^{^{\prime }}+(9\sin ^{2}t-3\cos t)\mu =0.$ The two
fundamental solutions are given by $\mu _{0}(t)=te^{3(1-\cos t)}$ and $\mu
_{1}(t)=e^{3(1-\cos t)},$ and also by $\gamma _{0}(t)=1/2t$ and $\mu (t)=\mu
(0)e^{3(1-\cos t)}\left[ 2\alpha (0)t+1\right] $. The explicit solution of
the form (\ref{Anzats1}) will satisfy 
\begin{equation}
\left\vert \psi \left( x,t\right) \right\vert =\frac{1}{e^{3(1-\cos t)}\mu
(0)(2\alpha (0)t+1)},
\end{equation}%
showing finite-time blow-up at $T^{^{{}}\ast }=-1/2\alpha (0)$. Again, this
result agrees with the prediction of the theorem above since $\gamma
_{0}^{-1}(-\alpha (0))=$ $-1/2\alpha (0).$
\end{example}

\section{Soliton Solutions for a Generalized Variable-Coefficient NLS Using
Ermakov's System}

In this section we present a family of Schr\"{o}dinger-type equations
admiting soliton solutions for (\ref{NLSVC}). Using a multiparameter
solution for the Ermakov system, see section 2, we present bright and
dark-type solitons for (\ref{NLSVC}), extending the results presented in 
\cite{Sua:Sus} where a Riccati system was used. We will use Lemma 3.
Further, by the use of these multiparameters, the solutions can be periodic
with bending propagation \ as in \cite{Mah:Su:Sus}\textbf{\ }for the
paraxial wave equation. We proceed to prove our second main result.

\begin{theorem}[Construction of solitons using Ermakov's system]
The nonlinear Schr\"{o}dinger equation with variable coefficients of the
form 
\begin{eqnarray}
i\psi _{t} &=&-a\left( t\right) \psi _{xx}+B\left( t\right) x^{2}\psi
-ic\left( t\right) x\psi _{x}-id\left( t\right) \psi  \label{ErmaNLS} \\
&&-M\left( t\right) x\psi +ig\left( t\right) \psi _{x}+L(t)\psi
+h(t)\left\vert \psi \right\vert ^{2}\psi  \notag
\end{eqnarray}%
has a soliton-type solution of the form%
\begin{equation}
\psi _{y}(t,x)=\frac{F(\beta (t)x+2\gamma (t)y+\varepsilon (t))}{\sqrt{\mu
(t)}}e^{i(\alpha \left( t\right) x^{2}+\beta \left( t\right) xy+\gamma
\left( t\right) y^{2}+\delta \left( t\right) x+\varepsilon \left( t\right)
y+\kappa \left( t\right) +\xi (t))},  \label{solitontype}
\end{equation}%
($y$ is a parameter) where $F$ satisfies%
\begin{equation}
F^{\prime \prime }=-\xi _{0}F+h_{0}F^{3},  \label{Elliptic}
\end{equation}%
and the following balance between coefficients (using (\ref{SErma11})-(\ref%
{SErma18})) has been imposed:%
\begin{eqnarray}
B\left( t\right) &=&b(t)-c_{0}a(t)\beta ^{4}(t),  \label{B(t)} \\
M(t) &=&f(t)+2c_{0}a(t)\beta ^{3}(t)\varepsilon \left( t\right) ,
\label{M(t)} \\
L(t) &=&c_{0}a(t)\beta ^{2}(t)\varepsilon ^{2}\left( t\right) ,  \label{L(t)}
\\
h(t) &=&h_{0}a(t)\beta ^{2}(t)\mu (t),  \label{h(t)} \\
\xi \left( t\right) &=&\xi _{0}(\gamma (t)-\gamma (0)).  \label{zeta(t)}
\end{eqnarray}
\end{theorem}

\begin{proof}
We look for a solution of the form 
\begin{equation}
\psi _{y}=A_{y}(x,t)e^{iS_{y}(x,t)}  \label{anzatsSoli}
\end{equation}%
with $S_{y}(t,x)=\alpha \left( t\right) x^{2}+\beta \left( t\right)
xy+\gamma \left( t\right) y^{2}+\delta \left( t\right) x+\varepsilon \left(
t\right) y+\kappa \left( t\right) +\xi (t),$ and $y$ is a parameter (we omit
the subindex $y$ in calculations).

Replacing (\ref{anzatsSoli}) in (\ref{ErmaNLS}) and assuming $A\geq 0,$ we
obtain 
\begin{eqnarray}
iA_{t}-AS_{t} &=&-a(t)A_{xx}-2ia(t)A_{x}S_{x}+a(t)S_{x}^{2}-ia(t)AS_{xx}
\label{3.10} \\
&&+b(t)x^{2}A-ic(t)xA_{x}+c(t)xAS_{x}-id(t)A  \notag \\
&&-f(t)xA+ig(t)A_{x}-g(t)AS_{x}+h(t)A^{3}.  \notag
\end{eqnarray}%
Taking the complex part, we obtain 
\begin{equation}
A_{t}=-((4a\alpha +c)x+2a\beta y+2a\delta -g)A_{x}-(2a\alpha +d)A,
\label{87}
\end{equation}%
taking the real part and equating coefficients as in \cite%
{Cor-Sot:Lop:Sua:Sus}. We thus obtain the nonlinear ODE 
\begin{equation}
aA_{xx}=\frac{d\xi }{dt}A+h(t)A^{3}.  \label{NODE}
\end{equation}%
Now, using (\ref{B(t)})-(\ref{L(t)}), we will obtain the Ermakov-type system
from (\ref{3.10}) 
\begin{align}
& \frac{d\alpha }{dt}+b+2c\alpha +4a\alpha ^{2}=c_{0}a\beta ^{4}, \\
& \frac{d\beta }{dt}+\left( c+4a\alpha \right) \beta =0, \\
& \frac{d\gamma }{dt}+a\beta ^{2}=0, \\
& \frac{d\delta }{dt}+\left( c+4a\alpha \right) \delta =f+2cg+2c_{0}a\beta
^{3}\varepsilon , \\
& \frac{d\varepsilon }{dt}=\left( g-2a\delta \right) \beta ,  \label{Ebeta}
\\
& \frac{d\kappa }{dt}=g\delta -a\delta ^{2}+c_{0}a\beta ^{2}\varepsilon ^{2}
\\
& \alpha \left( t\right) =-\frac{1}{4a\left( t\right) }\frac{\mu ^{\prime
}\left( t\right) }{\mu \left( t\right) }-\frac{d\left( t\right) }{2a\left(
t\right) }.
\end{align}%
The solution for this system is given by Lemma 3 in the appendix. Therefore,
we have obtained an explicit expression for $S_{y}(x,t)$ in (\ref{anzatsSoli}%
).

We proceed to find an expression for $A(x,t).$ Using (\ref{Ebeta}), we can
transform (\ref{87}) into 
\begin{equation}
A_{t}+\left( -\frac{\dot{\beta}}{\beta }x+2a\beta y-\frac{\dot{\varepsilon}}{%
\beta }\right) A_{x}+\frac{\dot{\mu}}{2\mu }A=0
\end{equation}%
and look for a solution for (\ref{NODE}) of the form%
\begin{equation}
A(x,t)=\frac{1}{\sqrt{\mu }}F(z),\qquad z=C_{0}(t)x+C_{1}(t)y+C_{2}(t).
\label{Anzats}
\end{equation}%
Anzats\ (\ref{Anzats}) and Ermakov's system guide us to choose $%
C_{0}(t)=\beta (t),$ $C_{1}(t)=\gamma (t)$ and $C_{2}(t)=\varepsilon (t),$
and therefore, (\ref{NODE}) becomes%
\begin{equation}
F_{zz}=\frac{\dot{\xi}}{\beta ^{2}(t)a(t)}F+\frac{h(t)}{\mu (t)\beta
^{2}(t)a(t)}F^{3}.  \label{Painllevenlode}
\end{equation}%
From here we obtain conditions (\ref{h(t)})-(\ref{zeta(t)}), and $A$ would
be given explicitly by%
\begin{equation*}
A(x,t)=\frac{F(\beta (t)x+2\gamma (t)y+\varepsilon (t))}{\sqrt{\mu (t)}}.
\end{equation*}%
\ \ 
\end{proof}

\begin{remark}
\bigskip Equation (\ref{Elliptic}) presents the following classical {%
nonlinear wave configurations, see \cite{Sua:Sus} and references therein.}

{\ }If $h_{0}<0$ 
\begin{eqnarray}
F\left( z\right) &=&\left( \frac{-\xi _{0}+\sqrt{\xi _{0}^{2}-2C_{0}h_{0}}}{%
-h_{0}}\right) ^{1/2}  \label{cn} \\
&&\times \text{cn}\left( \left( \xi _{0}^{2}-2C_{0}h_{0}\right)
^{1/4}z,\left( \frac{-\xi _{0}+\sqrt{\xi _{0}^{2}-2C_{0}h_{0}}}{2\sqrt{\xi
_{0}^{2}-2C_{0}h_{0}}}\right) ^{1/2}\right) ,  \notag
\end{eqnarray}%
then cn$\left( u,k\right) $ is a Jacobi elliptic function. A familiar
special case is obtained with $C_{0}=0,$ the \textit{bright} soliton:%
\begin{equation}
F\left( z\right) =\sqrt{\frac{-2\xi _{0}}{-h_{0}}}\frac{1}{\cosh \left( 
\sqrt{-\xi _{0}}z\right) }  \label{brightsoliton}
\end{equation}%
when cn$\left( u,1\right) =1/\cosh u.$

If $\xi _{0}>0$%
\begin{eqnarray}
F\left( z\right) &=&\left( \frac{\xi _{0}+\sqrt{\xi _{0}^{2}-2C_{0}h_{0}}}{%
h_{0}}\right) ^{1/2}  \label{sn} \\
&&\times \text{sn}\left( \left( \frac{C_{0}h_{0}}{\xi _{0}+\sqrt{\xi
_{0}^{2}-2C_{0}h_{0}}}\right) z,\left( \frac{-\xi _{0}-\sqrt{\xi
_{0}^{2}-2C_{0}h_{0}}}{-\xi _{0}+\sqrt{\xi _{0}^{2}-2C_{0}h_{0}}}\right)
^{1/2}\right) ,  \notag
\end{eqnarray}%
then sn$\left( u,k\right) $ is a Jacobi elliptic function. Another familiar
case is obtained with $C_{0}=\xi _{0}^{2}/\left( 2h_{0}\right) $ , the 
\textit{dark} soliton:%
\begin{equation}
F\left( z\right) =\sqrt{\frac{\xi _{0}}{h_{0}}}\tanh \left( \sqrt{\frac{\xi
_{0}}{2}}z\right)  \label{darksoliton}
\end{equation}%
sn$\left( u,1\right) =\tanh u.$
\end{remark}

\subsection{Family of solutions}

The following family of equations ($h(0),$ $\mu (0),$ $\alpha (0),$ $\beta
(0)$, $\gamma (0),$ $\delta (0)$, $\varepsilon (0)$ and $\kappa (0)$ are
parameters$)$%
\begin{equation}
i\psi _{t}=-\frac{1}{2}\psi _{xx}+(1-\beta ^{4})x^{2}\psi -2\beta
^{3}(t)\varepsilon (t)x\psi +\beta ^{2}(t)\varepsilon ^{2}(t)\psi +\frac{%
h(0)\beta ^{2}(0)}{\sqrt{\beta ^{4}(0)\sin ^{2}t+\left( 2\alpha (0)\sin
t+\cos t\right) ^{2}}}\left\vert \psi \right\vert ^{2}\psi ,  \label{Familly}
\end{equation}%
with

\begin{equation}
\beta \left( t\right) =\frac{\beta (0)}{\sqrt{\beta ^{4}(0)\sin ^{2}t+\left(
2\alpha (0)\sin t+\cos t\right) ^{2}}},
\end{equation}

\begin{equation}
\varepsilon \left( t\right) =\frac{\varepsilon (0)\left( 2\alpha (0)\sin
t+\cos t\right) -\beta (0)\delta (0)\sin t}{\sqrt{\beta ^{4}(0)\sin
^{2}t+\left( 2\alpha (0)\sin t+\cos t\right) ^{2}}},
\end{equation}%
admits a family of soliton solutions given by%
\begin{equation}
\psi _{y}(x,t)=\frac{F(\beta (t)x+2\gamma (t)y+\varepsilon (t))}{\sqrt{\mu
(t)}}e^{i(\alpha \left( t\right) x^{2}+\beta \left( t\right) xy+\gamma
\left( t\right) y^{2}+\delta \left( t\right) x+\varepsilon \left( t\right)
y+\kappa \left( t\right) +\xi (t))}  \label{Soliton type}
\end{equation}%
where $F$ satisfies (\ref{Elliptic}) and 
\begin{align}
\mu (t)& =\mu (0)\sqrt{\beta ^{4}(0)\sin ^{2}t+\left( 2\alpha (0)\sin t+\cos
t\right) ^{2}} \\
& \alpha \left( t\right) =\frac{\alpha (0)\cos 2t+\sin 2t\left( \beta
^{4}(0)+4\alpha ^{2}(0)-1\right) /4}{\beta ^{4}(0)\sin ^{2}t+\left( 2\alpha
(0)\sin t+\cos t\right) ^{2}}, \\
& \gamma \left( t\right) =\gamma (0)-\frac{1}{2}\arctan \frac{\beta
^{2}(0)\sin t}{2\alpha (0)\sin t+\cos t}, \\
& \delta \left( t\right) =\frac{\delta (0)\left( 2\alpha (0)\sin t+\cos
t\right) +\varepsilon (0)\beta ^{3}(0)\sin t}{\beta ^{4}(0)\sin ^{2}t+\left(
2\alpha (0)\sin t+\cos t\right) ^{2}}, \\
& \kappa \left( t\right) =\kappa (0)+\sin ^{2}t\frac{\varepsilon (0)\beta
^{2}(0)\left( \alpha (0)\varepsilon (0)-\beta (0)\delta (0)\right) -\alpha
(0)\delta ^{2}(0)}{\beta ^{4}(0)\sin ^{2}t+\left( 2\alpha (0)\sin t+\cos
t\right) ^{2}} \\
& \qquad \quad +\frac{1}{4}\sin 2t\frac{\varepsilon ^{2}(0)\beta
^{2}(0)-\delta ^{2}(0)}{\beta ^{4}(0)\sin ^{2}t+\left( 2\alpha (0)\sin
t+\cos t\right) ^{2}},  \notag \\
\xi \left( t\right) & =\xi _{0}(\gamma (t)-\gamma (0)).
\end{align}

Now we show that this family contains classic soliton examples.

\begin{example}
(Classical bright soliton) The equation (\ref{Familly})\ with $\delta (0)=0,$
$h(0)=-2,$ $\beta (0)=1,$ $\mu (0)=1$ and $\alpha (0)=\gamma (0)=\varepsilon
(0)=\kappa (0)=0$ becomes the classical NLS 
\begin{equation}
i\psi _{t}=-\frac{1}{2}\psi _{xx}-2\left\vert \psi \right\vert ^{2}\psi ,
\end{equation}

which admits the bright soliton of the form ($y$ is velocity)%
\begin{equation}
\left\vert \psi _{y}(x,t)\right\vert ^{2}=\sech^{2}(x-ty).
\end{equation}
\end{example}

\begin{example}
(Classical dark soliton) The equation (\ref{Familly})\ with $\delta (0)=0,$ $%
h(0)=2,$ $\beta (0)=1,$ $\mu (0)=1$ and $\alpha (0)=\gamma (0)=\varepsilon
(0)=\kappa (0)=0$ becomes the classical NLS 
\begin{equation}
i\psi _{t}=-\frac{1}{2}\psi _{xx}+2\left\vert \psi \right\vert ^{2}\psi ,
\end{equation}

which admits the dark soliton of the form ($y$ is velocity)%
\begin{equation}
\left\vert \psi _{y}(x,t)\right\vert ^{2}=\tanh ^{2}(x-ty).
\end{equation}
\end{example}

The following examples show periodic solutions with bending propagation.

\begin{example}
(Bright-type soliton) The equation (\ref{Familly})\ with $\delta (0)$ a
parameter$,$ $h(0)=-2,$ $\beta (0)=2/3,$ $\mu (0)=1$ and $\alpha (0)=\gamma
(0)=\varepsilon (0)=\kappa (0)=0$ becomes%
\begin{equation}
i\psi _{t}=-\frac{1}{2}\psi _{xx}+(1-\beta ^{4})x^{2}\psi -2\beta
^{3}(t)\varepsilon (t)x\psi +\beta ^{2}(t)\varepsilon ^{2}(t)\psi -\frac{8}{9%
\sqrt{\frac{16}{81}\sin ^{2}t+\cos ^{2}t}}\left\vert \psi \right\vert
^{2}\psi ,  \label{bright family}
\end{equation}%
with%
\begin{eqnarray}
\beta \left( t\right) &=&\frac{2}{3\sqrt{\frac{16}{81}\sin ^{2}t+\cos ^{2}t}}%
,  \label{betaexamplebright} \\
\varepsilon \left( t\right) &=&\frac{-2\delta (0)\sin t}{3\sqrt{\frac{16}{81}%
\sin ^{2}t+\cos ^{2}t}}.  \label{epsilonexamplebright}
\end{eqnarray}

(\ref{bright family}) admits a bright-type soliton solution with absolute
value of the form%
\begin{equation}
\left\vert \psi _{y}(x,t)\right\vert ^{2}=\frac{\sech^{2}(\beta (t)x+2\gamma
(t)y+\varepsilon (t))}{\sqrt{\beta ^{4}(0)\sin ^{2}t+\left( 2\alpha (0)\sin
t+\cos t\right) ^{2}}}.  \label{bright}
\end{equation}%
We observe that $\delta (0)$ produces a bending effect, see Figure \ref%
{particular1}.

\begin{figure}[h!]
\centering
\subfigure[]{\includegraphics[scale=0.55]{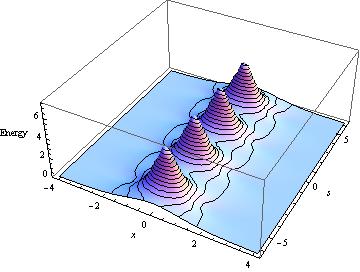}} \subfigure[]{%
\includegraphics[scale=0.55]{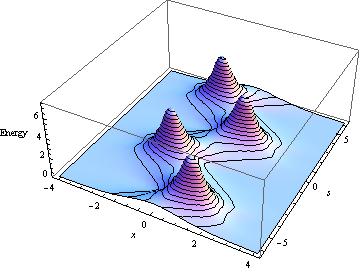}}
\caption{(a) Solution for $(\protect\ref{bright family})\text{ with }\protect%
\delta (0)=0.$ (b) Solution for $(\protect\ref{bright family})\text{ with }%
\protect\delta (0)=1.$}
\label{particular1}
\end{figure}
\end{example}

\begin{example}
(Dark-type soliton) The equation (\ref{Familly})\ with $\delta (0)$ a
parameter$,$ $h(0)=2,$ $\beta (0)=2/3,$ $\mu (0)=1$ and $\alpha (0)=\gamma
(0)=\varepsilon (0)=\kappa (0)=0$ becomes%
\begin{equation}
i\psi _{t}=-\frac{1}{2}\psi _{xx}+(1-\beta ^{4})x^{2}\psi -2\beta
^{3}(t)\varepsilon (t)x\psi +\beta ^{2}(t)\varepsilon ^{2}(t)\psi +\frac{8}{9%
\sqrt{\frac{16}{81}\sin ^{2}t+\cos ^{2}t}}\left\vert \psi \right\vert
^{2}\psi  \label{Darkfamily}
\end{equation}

and admits a dark-type soliton solution with absolute value of the form%
\begin{equation}
\left\vert \psi _{y}(x,t)\right\vert ^{2}=\frac{\tanh ^{2}(\beta
(t)x+2\gamma (t)y+\varepsilon (t))}{\sqrt{\beta ^{4}(0)\sin ^{2}t+\left(
2\alpha (0)\sin t+\cos t\right) ^{2}}}.  \label{dark}
\end{equation}

We observe again that $\delta (0)$ produces a bending effect, see Figure \ref%
{particular2}.

\begin{figure}[h!]
\centering
\subfigure[]{\includegraphics[scale=0.55]{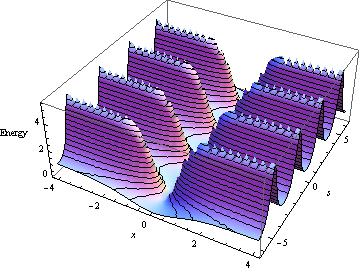}} \subfigure[]{%
\includegraphics[scale=0.55]{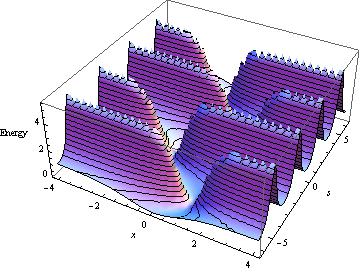}}
\caption{(a) Solution for $(\protect\ref{Darkfamily})\text{ with }\protect%
\delta (0)=0.$ (b) Solution for $(\protect\ref{Darkfamily})\text{ with }%
\protect\delta (0)=1.$ }
\label{particular2}
\end{figure}
\end{example}

\section{Dynamics of Explicit Solutions Through Parameters}

In this section using Lemmas 1 and 2, we will give examples of nonautonomous
nonlinear Schr\"{o}dinger equations with multiparameter solutions. These toy
examples show the control of the dynamics of the solutions as an application
of our multiparameter approach. We have prepared a Mathematica file as
supplemental material for this Section. In this file all the solutions for
this Section are verified. Also, all the formulas from the appendix have
been verified previously in \cite{Ko:Su;su}.In the following two examples we
use Lemma 2 from the appendix.

\subsection{Dynamics of the bright soliton: Bending propagation}

Consider the nonautonomous nonlinear Schr\"{o}dinger equation\ 
\begin{equation}
i\psi _{t}+\frac{1}{2}\psi _{xx}+ix\tanh t\psi _{x}+i\cosh t\psi +\frac{%
e^{2\sinh t}\sech t}{2\sinh t+\cosh t}|\psi |^{2}\psi =0,\quad x\in \mathbb{R%
},\hspace{0.3cm}t>0.  \label{General1}
\end{equation}%
\ Then, the characteristic equation and its solution are given by\ 
\begin{equation*}
\mu ^{\prime \prime }-(4\cosh t-2\tanh t)\mu ^{\prime }+(4\cosh ^{2}t-6\sinh
t)\mu =0.
\end{equation*}%
\begin{equation*}
\mu (t)=(2+\coth t)e^{2\sinh t}\tanh t.
\end{equation*}%
\ By Lemma 2 (\ref{General1}) can be reduced to:\ 
\begin{equation}
iu_{\tau }-u_{\xi \xi }-2|u|^{2}u=0,\quad \xi =\beta (t)x+\varepsilon
(t),\quad \tau =\gamma (t),  \label{standard1}
\end{equation}%
\ where the general solution of the Riccati system associated is:\ 
\begin{eqnarray*}
\alpha (t) &=&\dfrac{\csch t\sech t}{2+\coth t},\quad \beta (t)=\dfrac{\csch %
t}{2+\coth t},\quad \gamma (t)=\gamma (0)-\dfrac{1}{4+2\coth t}, \\
&& \\
\delta (t) &=&\dfrac{\delta (0)\csch t}{2+\coth t},\quad \varepsilon
(t)=\varepsilon (0)-\dfrac{\delta (0)}{2+\coth t},\quad \kappa (t)=\kappa
(0)-\dfrac{\delta (0)^{2}}{4+2\coth t}.
\end{eqnarray*}%
In order to use the similarity transformation method we proceed to use the
familiar solution for (\ref{standard1}): $u(\tau ,\xi )=\sqrt{v}\sech(\sqrt{v%
}\xi )\exp (-iv\tau ),\quad v>0$.

Therefore, a solution for the Schr$\ddot{\mbox{o}}$dinger equation (\ref%
{General1}) is given by:%
\begin{eqnarray}
\psi (t,x) &=&\sqrt{\frac{v\coth t}{2+\coth t}}\sech\left[ \sqrt{v}\left( 
\frac{x\csch t-\delta (0)}{2+\coth t}+{\varepsilon (0)}\right) \right]
\label{g1} \\
&\times &\exp \left[ i\left( \frac{2x^{2}\csch t\sech t+2\delta (0)x\csch %
t-\delta (0)^{2}+v}{4+2\coth t}\right) \right]  \notag \\
&\times &\exp \left[ i\left( \kappa (0)-v\gamma (0)\right) -\sinh t\right] .
\notag
\end{eqnarray}

The dynamics of the solution (\ref{g1}) are shown in Figure \ref{f1}, where
it was possible to produce a change in the central axis of the bright
soliton for certain values $\delta (0)$ and $\varepsilon (0)$. 
\begin{figure}[h]
\centering
\subfigure[Bright soliton solution bended to the left: $v=1$,
$\delta(0)=-30$ and $\varepsilon(0)=0$.]{\includegraphics[scale=0.29]{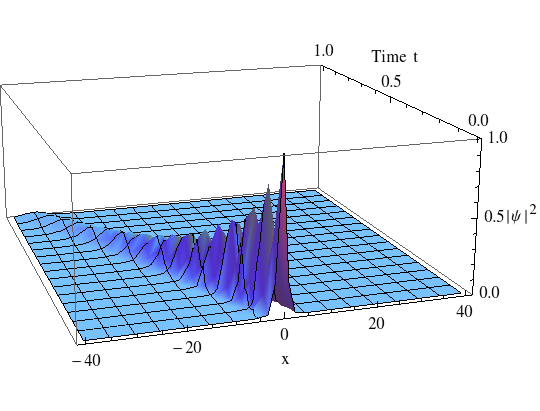}}
\subfigure[Centered bright soliton: $v=1$, $\delta
(0)=\frac{1}{6}$ and $\varepsilon(0)=-\frac{1}{6}$.] {\includegraphics[scale=0.29]{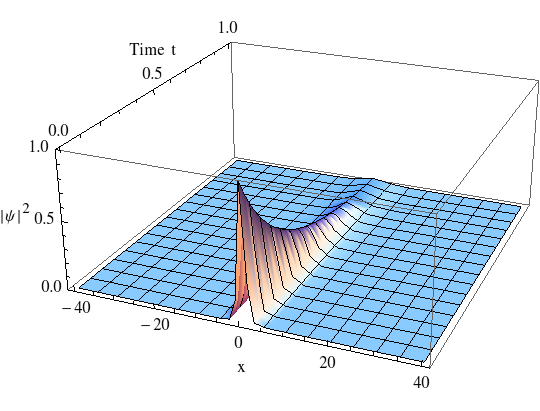}}
\subfigure[Bright soliton solution bended to the right: $v=1$, $\delta
(0)=0$ and $\varepsilon(0)=-10$.]{\includegraphics[scale=0.29]{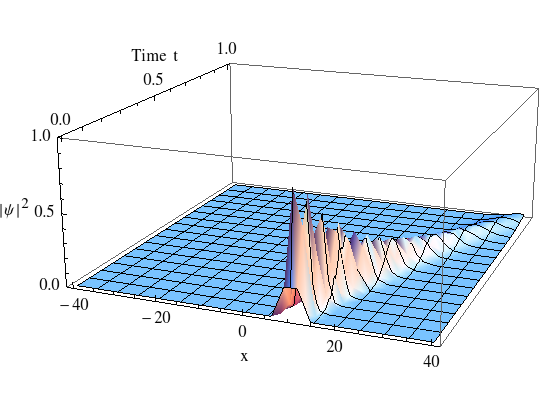}} 
\subfigure[Contour of the bright soliton solution bended to the
left.]{\includegraphics[scale=0.32]{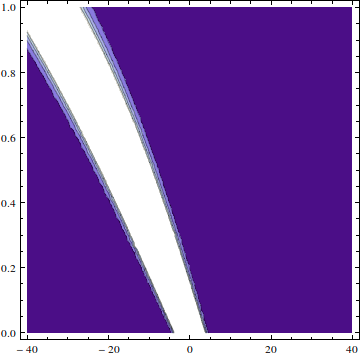}}\hspace{0.9cm} %
\subfigure[Contour for the centered bright soliton.]
{\includegraphics[scale=0.32]{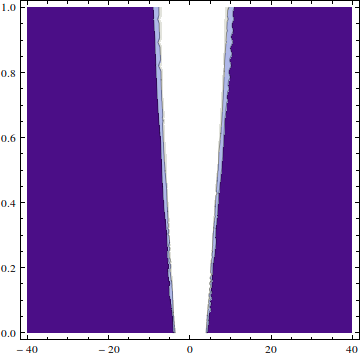}}\hspace{0.9cm} 
\subfigure[Contour
of the bright soliton solution bended to the
right.]{\includegraphics[scale=0.32]{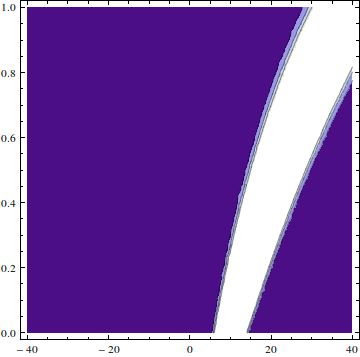}}
\caption{Control on the dynamics of the solution (\protect\ref{g1}) for the
equation (\protect\ref{General1}). }
\label{f1}
\end{figure}

\subsection{Dynamics of the dark soliton: Bending propagation}

Consider the nonautonomous nonlinear Schr\"{o}dinger equation\ 
\begin{equation}
i\psi _{t}=-\frac{1}{2}\cosh t\psi _{xx}+\frac{1}{2}\cosh t\psi
(x^{2}-i)-ix\cosh t\psi _{x}+\frac{4\cosh t}{1+\sinh t}|\psi |^{2}\psi
,\quad x\in \mathbb{R},\hspace{0.3cm}t>0;  \label{ds}
\end{equation}%
\ then, the characteristic equation and its solution are respectively\ 
\begin{equation*}
\mu ^{\prime \prime }-\tanh t\mu ^{\prime }=0,
\end{equation*}%
\begin{equation*}
\mu (t)=1+\sinh t.
\end{equation*}%
\ By Lemma 2 (\ref{ds}) can be reduced to\ 
\begin{equation}
iu_{\tau }+u_{\xi \xi }-2|u|^{2}u=0,\hspace{1cm}\xi =\beta (t)x+\varepsilon
(t),\quad \tau =\gamma (t).  \label{standard2}
\end{equation}%
\ The general solution of the corresponding Riccati system is given by\ 
\begin{eqnarray*}
\alpha (t) &=&-\frac{1}{2+2\csch t},\quad \beta (t)=\frac{1}{1+\sinh t}%
,\quad \gamma (t)=\frac{2}{1+\csch t}-\gamma (0), \\
&& \\
\delta (t) &=&\frac{\delta (0)}{1+\sinh t},\quad \varepsilon (t)=\varepsilon
(0)-\frac{2\delta (0)}{1+\csch t},\quad \kappa (t)=\kappa (0)-\frac{\delta
^{2}(0)}{2+2\csch t}.
\end{eqnarray*}%
In order to use the similarity transformation method we proceed to use the
familiar solution for (\ref{standard2}): $u(\tau ,\xi )=A\tanh \left( A\xi
\right) \exp (-2iA^{2}\tau ),$ $A\in \mathbb{R}.$ As a consequence, a
solution for the Schr$\ddot{\mbox{o}}$dinger equation (\ref{ds}) has the
following form:\ 
\begin{eqnarray}
\psi (t,x) &=&\frac{A}{\sqrt{1+\sinh t}}\tanh \left[ A\left( \frac{2x\csch %
t-2\delta (0)}{\csch t+1}+\varepsilon (0)\right) \right]  \label{g2} \\
\ &\times &\exp \left[ i\left( \frac{-x^{2}+2\delta (0)x\csch t-\delta
^{2}(0)-8A^{2}}{2+2\csch t}\right) \right]  \notag \\
\ &\times &\exp \left[ i\left( \kappa (0)+2A^{2}\gamma (0)\right) \right] . 
\notag
\end{eqnarray}%
Figure \ref{f2} describes the evolution in time of the solution (\ref{g2}).
Again for this case we have bending propagation. 
\begin{figure}[h]
\centering
\subfigure[Dark soliton solution bended to the left: $A=2$,
$\delta(0)=-1$ and $\varepsilon(0)=0$.]{\includegraphics[scale=0.28]{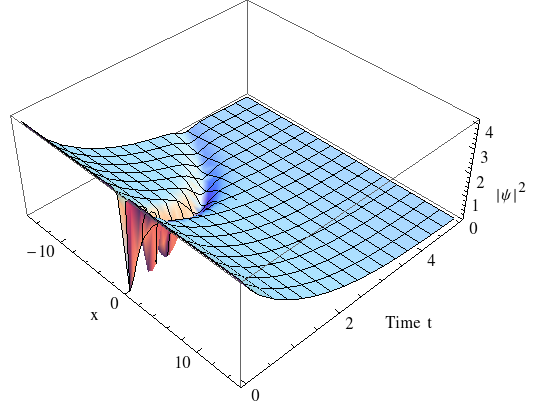}}
\subfigure[Centered dark soliton: $A=2$, $\delta
(0)=0$, $\varepsilon(0)=0$]{\includegraphics[scale=0.28]{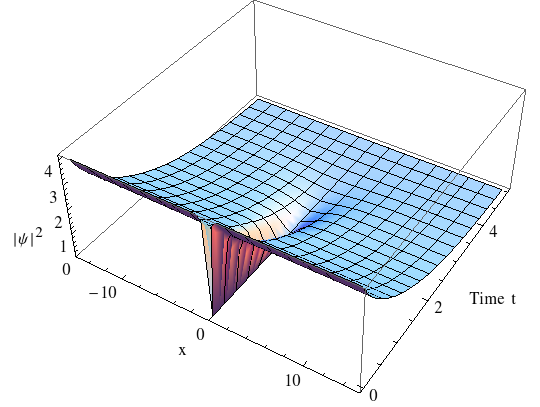}} 
\subfigure[Dark soliton solution bended to the right: $A=2$, $\delta
(0)=0$ and $\varepsilon(0)=-2$.]{\includegraphics[scale=0.28]{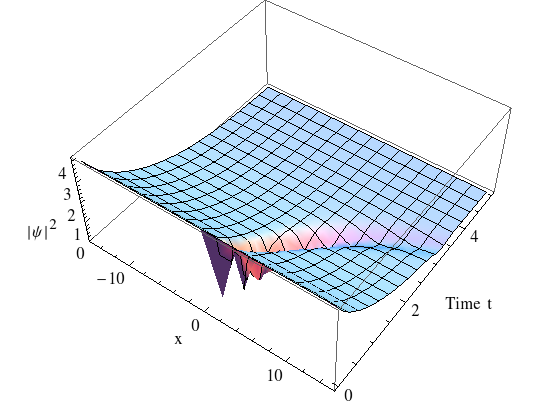}} 
\subfigure[Contour of the dark soliton solution bended to the
left.]{\includegraphics[scale=0.32]{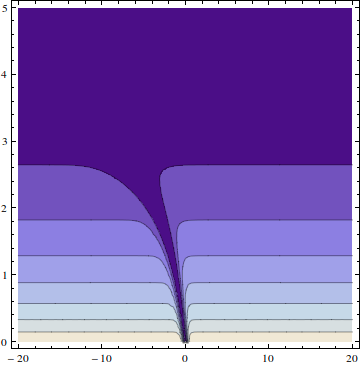}}\hspace{0.9cm} 
\subfigure[Contour of the centered  dark
soliton.]{\includegraphics[scale=0.32]{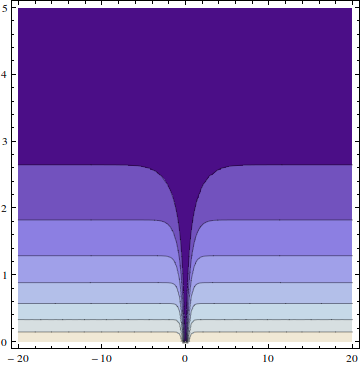}}\hspace{0.9cm} 
\subfigure[Contour of the dark soliton solution bended to the
right.]{\includegraphics[scale=0.32]{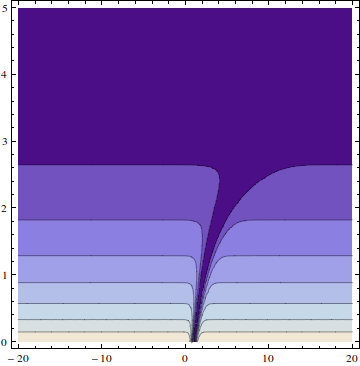}}
\caption{Control on the dynamics of the solution (\protect\ref{g2}) for the
equation (\protect\ref{ds}).}
\label{f2}
\end{figure}

The difficulty of applying Lemma 2 is solving the Riccati system. Next, we
present an alternative approach to deal with the Riccati system and see how
the dynamics of the solutions change with multiparameters.

\subsection{An alternative method to solve the coupled Riccati system and
applications to soliton solutions}

Assume the following conditions on the Riccati system (\ref{rica1})-(\ref%
{rica6}): $a(t)=-l_{0}$ with $l_{0}=\pm 1$, $\beta (t)=1$, $\tau (t)=t$ and $%
\varepsilon (t)=0$. Under this hypothesis one obtains the explicit formulas $%
\alpha (t)=l_{0}c(t)/4$, $\delta (t)=-l_{0}g(t)/2$ and $h(t)=-l_{0}\lambda
\mu (t)$, where the last expression shows the unique dependence of the
coefficient of the nonlinearity in terms of the characteristic function $\mu 
$. Furthermore one obtains the particular Riccati system\ 
\begin{equation}
\dfrac{dc}{dt}+c^{2}+4l_{0}b=0,  \label{ricafin1}
\end{equation}%
\begin{equation}
\dfrac{dg}{dt}+2l_{0}f+cg=0,  \label{ricafin2}
\end{equation}%
\begin{equation}
\dfrac{d\kappa }{dt}+\frac{l_{0}}{4}g^{2}=0,  \label{ricafin3}
\end{equation}%
\begin{equation}
\dfrac{d\mu }{dt}=(2d-c)\mu .  \label{ricafin4}
\end{equation}%
\ The solution of this system is given by\ 
\begin{equation}
\dfrac{dc}{dt}+c^{2}+4l_{0}b=0,  \label{opt1}
\end{equation}%
\begin{equation}
\alpha (t)=l_{0}\frac{c(t)}{4},\quad \delta (t)=-l_{0}\frac{g(t)}{2},\quad
h(t)=-l_{0}\lambda \mu (t),  \label{opt2}
\end{equation}%
\begin{equation}
\kappa (t)=\kappa (0)-\frac{l_{0}}{4}\int_{0}^{t}g^{2}(z)dz,  \label{opt3}
\end{equation}%
\begin{equation}
\mu (t)=\mu (0)\mbox{exp}\left( \int_{0}^{t}(2d(z)-c(z))dz\right) ,\quad \mu
(0)\neq 0,  \label{opt4}
\end{equation}%
\begin{equation}
g(t)=g(0)-2l_{0}\mbox{exp}\left( -\int_{0}^{t}c(z)dz\right) \int_{0}^{t}%
\mbox{exp}\left( \int_{0}^{z}c(y)dy\right) f(z)dz.  \label{opt5}
\end{equation}

Further, with these restrictions we have a way to construct transformations
that allow us to prove uniqueness of the solutions.

\begin{lemma}
\label{P1} Suppose that $h(t)=-l_{0}\lambda \mu (t)$ with $\lambda \in 
\mathbb{R}$, $l_{0}=\pm 1$ and that $c(t)$, $\alpha (t)$, $\delta (t)$, $%
\kappa (t)$, $\mu (t)$ and $g(t)$ satisfy the equations (\ref{opt1})-(\ref%
{opt5}). Then\ 
\begin{equation}
\psi (t,x)=\dfrac{1}{\sqrt{\mu (t)}}e^{i(\alpha (t)x^{2}+\delta (t)x+\kappa
(t))}u(t,x)  \label{tranpar1}
\end{equation}%
\ is a solution to the Cauchy problem for the nonautonomous nonlinear Schr$%
\ddot{\mbox{o}}$dinger equation\ 
\begin{equation}
i\psi _{t}=l_{0}\psi _{xx}+b(t)x^{2}\psi -ic(t)x\psi _{x}-id(t)\psi
-f(t)x\psi +ig(t)\psi _{x}+h(t)|\psi |^{2}\psi ,  \label{schpart1}
\end{equation}%
\begin{equation}
\psi (0,x)=\psi _{0}(x)  \label{schpart2}
\end{equation}%
\ if and only if $u(t,x)$ is a solution of the Cauchy problem for the
standard nonlinear Schr$\ddot{\mbox{o}}$dinger equation\ 
\begin{equation}
iu_{t}-l_{0}u_{xx}+l_{0}\lambda |u|^{2}u=0.  \label{schusu}
\end{equation}%
\begin{equation}
u(0,x)=\sqrt{\mu (0)}e^{-i(\alpha (0)x^{2}+\delta (0)x+\kappa (0))}\psi
_{0}(x).  \label{schusu1}
\end{equation}
\end{lemma}

The following proposition establishes the uniqueness of classical solutions
for the nonautonomous nonlinear Schr$\ddot{\mbox{o}}$dinger equation (\ref%
{schpart1}).

\begin{proposition}
Assume that equations (\ref{opt1})-(\ref{opt5}) are satisfied (corresponding
to $l_{0}=-1$). If $c(t)$, $d(t)$ and $f(t)\in C^{1}([-T,T])$ for some $T>0$
and $h(t)=\lambda \mu (t)$ with $\lambda \in \mathbb{R}$, then the Cauchy
problem for the nonlinear nonautonomous Schr$\ddot{\mbox{o}}$dinger equation%
\begin{equation}
i\psi _{t}=-\psi _{xx}+b(t)x^{2}\psi -ic(t)x\psi _{x}-id(t)\psi -f(t)x\psi
+ig(t)\psi _{x}+h(t)|\psi |^{2}\psi  \label{exi1}
\end{equation}%
\begin{equation}
\psi (0,x)=\psi _{0}(x)  \label{exi2}
\end{equation}%
\ has a unique classical solution in the space $L_{t}^{\infty
}L_{x}^{q}([-T,T]\times \mathbb{R})$ for $q=2,\infty $.
\end{proposition}

\begin{proof}
Let's consider $\psi ^{1},\psi ^{2}\in C_{t,x}^{2}([-T,T]\times \mathbb{R})$
classical solutions for the Cauchy problem (\ref{exi1})-(\ref{exi2}) in the
space $L_{t}^{\infty }L_{x}^{q}([-T,T]\times \mathbb{R})$ for $q=2,\infty $.
By Lemma \ref{P1} and the conditions in the coefficients $c(t)$, $d(t)$, $%
f(t)$ and $h(t)$, we have 
\begin{equation}
u^{j}(t,x)=\sqrt{\mu (t)}e^{-i(\alpha (t)x^{2}+\delta (t)x+\kappa (t))}\psi
^{j}(t,x)\in C_{t,x}^{2}([-T,T]\times \mathbb{R})\hspace{0.5cm}\mbox{with}%
\hspace{0.5cm}j=1,2,  \label{inv}
\end{equation}%
which are classical solutions for the Cauchy problem (\ref{schusu})-(\ref%
{schusu1}) on $[-T,T]$, and initial condition 
\begin{equation*}
u^{j}(0,x)=\sqrt{\mu (0)}e^{-i(\alpha (0)x^{2}+\delta (0)x+\kappa (0))}\psi
_{0}(x)\hspace{0.5cm}\mbox{for}\hspace{0.5cm}j=1,2.
\end{equation*}%
Therefore, for each $j=1,2$ we have 
\begin{equation}
\Vert u^{j}\Vert _{L_{t}^{\infty }L_{x}^{q}([-T,T]\times \mathbb{R})}\leq
M\Vert \psi ^{j}\Vert _{L_{t}^{\infty }L_{x}^{q}([-T,T]\times \mathbb{R})},%
\hspace{1cm}q=2,\infty ,  \notag
\end{equation}%
where $M>0$ is the maximum for $\mu (t)$ in the interval $[-T,T]$. Using the
classical uniqueness result given in \cite{Tao}, we have $u^{1}=u^{2}$, and
so, the final result is obtained by multiplying equation (\ref{inv}) by the
factor $e^{i(\alpha (t)x^{2}+\delta (t)x+\kappa (t))}/\sqrt{\mu (t)}$.
\end{proof}

Now we want to see how the dynamics of the solutions of the nonlinear Schr$%
\ddot{\mbox{o}}$dinger equation (\ref{schpart1}) change when the parameters
of dissipation, $d(t)$, and the nonlinear term, $h(t)$ change. We will use
the alternative Riccati system (\ref{ricafin1})-(\ref{ricafin4}) and
therefore Lemma 1.

\subsection{Perturbations of the bright soliton: Competition between
dissipation and nonlinearity}

Let's consider the nonautonomous nonlinear Schr$\ddot{\mbox{o}}$dinger
equation\ 
\begin{equation}
i\psi _{t}=-\psi _{xx}+\frac{x^{2}}{4}\left( \sin ^{2}t-\cos t\right) \psi
+ix\sin t\psi _{x}-id(t)\psi +h(t)|\psi |^{2}\psi ,\quad t,x\in \mathbb{R}.
\label{sch1}
\end{equation}%
\ Then the functions $\alpha $, $\delta $ and $\kappa $ are respectively:\ 
\begin{equation*}
\alpha (t)=\frac{\sin t}{4},\quad \quad \delta (t)=0,\quad \quad \kappa
(t)=\kappa (0).
\end{equation*}%
\ We will construct explicit solutions for (\ref{sch1}) using Lemma 1 and
one of the solutions for the NLS equation \ 
\begin{equation*}
iu_{t}+u_{xx}+3|u|^{2}u=0,\quad t,x\in \mathbb{R},
\end{equation*}%
that is given by%
\begin{equation*}
u(t,x)=\sqrt{-\frac{2v}{3}}\sech(\sqrt{-v}x)\exp (-vit),\quad v<0.
\end{equation*}

\subsubsection{Periodic Solutions for (\protect\ref{sch1}) with $d(t)=\sin t$
and $h(t)=-3e^{3-3\cos t}$}

In this case one obtains a solution of the following form:\ 
\begin{eqnarray}
\psi (t,x) &=&\exp \left[ \frac{3}{2}\left( \cos t-1\right) +i\left( \frac{%
x^{2}}{4}\sin t+\kappa (0)-vt\right) \right]  \label{Periodic1} \\
\ &\times &\sqrt{-\frac{2v}{3}}\sech(\sqrt{-v}x),  \notag
\end{eqnarray}

see Figure \ref{f3}. 
\begin{figure}[h]
\centering
\subfigure[Solution with values $v=-2$
and $\protect\kappa(0)=0$.]{\includegraphics[scale=0.38]{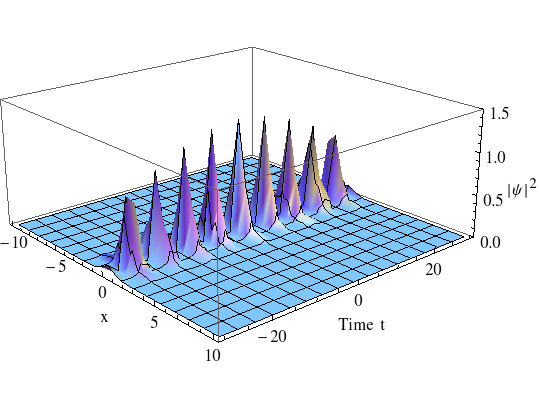}} 
\subfigure[Contour of the
solution.]{\includegraphics[scale=0.38]{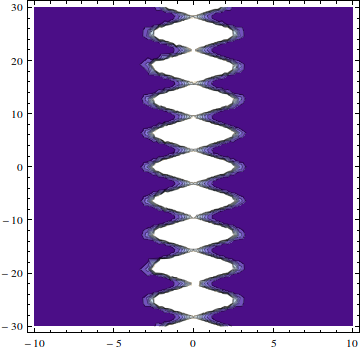}}
\caption{Dynamics of the solution (\protect\ref{Periodic1}) for equation (%
\protect\ref{sch1}) with $d(t)=\sin t$ and $h(t)=-3e^{3-3\cos t}$.}
\label{f3}
\end{figure}

\subsubsection{Solutions with fast decay (\protect\ref{sch1}) with $%
d(t)=(4t-\sin t)/2$ and $h(t)=-3e^{2t^{2}}$}

The choice of the parameters allows us to construct a solution with fast
decay for large values of time:\ 
\begin{equation}
\psi (t,x)=\sqrt{-\frac{2v}{3}}\sech(\sqrt{-v}x)\exp \left[ i\left( \frac{%
x^{2}}{4}\sin t+\kappa (0)-vt\right) -t^{2}\right] ,  \label{FastDecay}
\end{equation}

see Figure \ref{f4}. 

\begin{figure}[h!]
\centering
\subfigure[Solution with values $v=-2$
and $\protect\kappa(0)=0$.]{\includegraphics[scale=0.38]{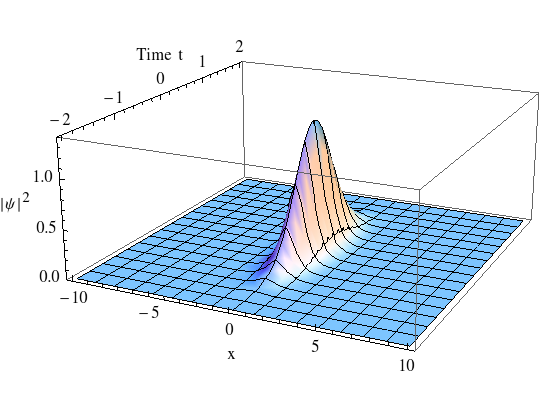}} 
\subfigure[Contour of the
solution.]{\includegraphics[scale=0.38]{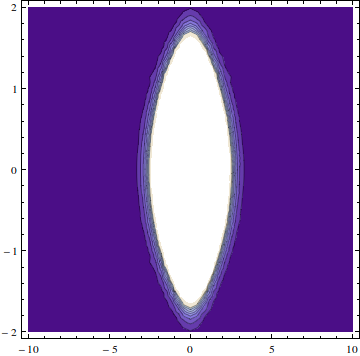}}
\caption{Dynamics of the solution (\protect\ref{FastDecay}) for equation (%
\protect\ref{sch1}) with $d(t)=(4t-\sin t)/2$ and $h(t)=-3e^{2t^{2}}$.}
\label{f4}
\end{figure}

\subsection{Perturbations of the Peregrine soliton}

We are interested to see how the parameters change the dynamics of solutions
for (\ref{schpart1}), and for this end we consider \ 
\begin{equation}
i\psi _{t}=-\psi _{xx}+x^{2}\left( t^{2}-1/2\right) \psi +2it\psi
_{x}(x+e^{t^{2}})-id(t)\psi -e^{t^{2}}x\psi +h(t)|\psi |^{2}\psi ,\quad
t,x\in \mathbb{R}.  \label{sch2}
\end{equation}%
\ As before we can find explicitly 
\begin{equation*}
\alpha (t)=t/2,\quad \delta (t)=te^{t^{2}},\quad \kappa (t)=\kappa
(0)+e^{2t^{2}}\left( 2t-\sqrt{2}D(\sqrt{2}t)\right) /8,
\end{equation*}%
\ where $D(t)=e^{-t^{2}}\int_{0}^{t}e^{z^{2}}dz$ is the Dawson function. We
will construct explicit solutions for (\ref{sch2}) using Lemma 2 to reduce
it to \ 
\begin{equation}
iu_{t}+u_{xx}+2|u|^{2}u=0,\quad t,x\in \mathbb{R},  \label{Pere}
\end{equation}%
and one of the solutions for (\ref{Pere}) is%
\begin{equation}
u(t,x)=A\exp (2iA^{2}t)\left( \frac{3+16iA^{2}t-16A^{4}t^{2}-4A^{2}x^{2}}{%
1+16A^{4}t^{2}+4A^{2}x^{2}}\right) ,\quad A\in \mathbb{R}.  \notag
\end{equation}

\subsubsection{Peregrine-type soliton for (\protect\ref{sch2}) with $%
d(t)=\tanh t-t$ and $h(t)=-8\cosh ^{2}t$}

The correct choice of the parameters $d(t)$ and $h(t)$ allows us to
construct solutions with properties similar to those of the classical
Peregrine soliton, as can be seen in Figure \ref{part3}. The solution for
this case will be given by\ 
\begin{eqnarray}
\psi (t,x) &=&\exp \left[ i\left( \frac{t}{2}x^{2}+te^{t^{2}}x+\kappa (0)+%
\frac{1}{8}e^{2t^{2}}\left( 2t-\sqrt{2}D(\sqrt{2}t)\right) \right) \right]
\label{Peregrine1} \\
\ &\times &\frac{A}{2}\exp (2A^{2}t)\left( \frac{%
3+16iA^{2}t-16A^{4}t^{2}-4A^{2}x^{2}}{1+16A^{4}t^{2}+4A^{2}x^{2}}\right) %
\sech t.  \notag
\end{eqnarray}


\begin{figure}[h!]
\centering
\subfigure[Profile of Peregrine soliton at
$t=0$.]{\includegraphics[scale=0.39]{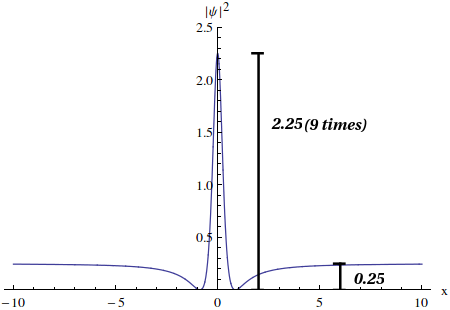}} 
\subfigure[Profile of
Peregrine soliton at $x=0$.]{\includegraphics[scale=0.39]{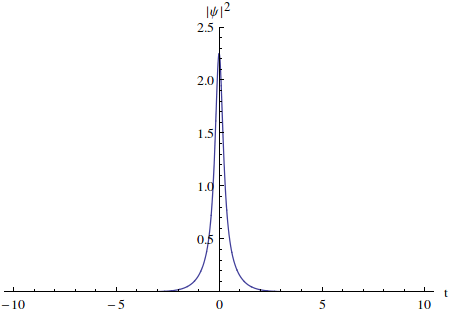}} 
\subfigure[3D view of Peregrine
soliton.]{\includegraphics[scale=0.38]{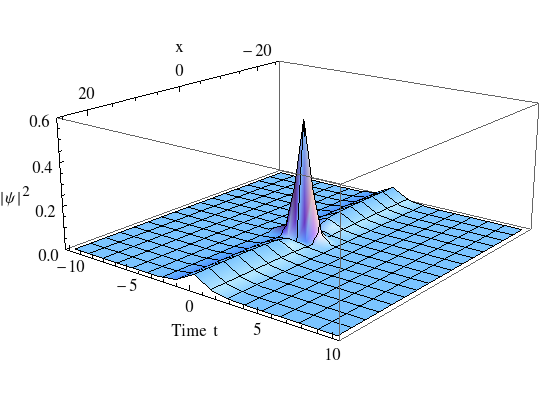}}
\caption{Peregrine soliton solution for (\protect\ref{Peregrine1}) with $A= 
\foreignlanguage{english}{0.5}$ and $\protect\kappa (0)=0$ for (\protect\ref%
{sch2}) with $d(t)=\tanh (t)-t$ and $h(t)=-8\cosh ^{2}(t)$.}
\label{part3}
\end{figure}

\subsubsection{Dynamics of the solution for (\protect\ref{sch2}) with $%
d(t)=-(\sin 2t+t)$ and $h(t)=-2e^{-2\sin ^{2}t}$}

Here we see how the solutions are perturbated, see Figure \ref{part2} in the
appendix. The solution is given by\ 
\begin{eqnarray}
\psi (t,x) &=&A\exp \left[ i\left( \frac{t}{2}x^{2}+te^{t^{2}}x+\kappa (0)+%
\frac{1}{8}e^{2t^{2}}\left( 2t-\sqrt{2}D(\sqrt{2}t)\right) \right) \right]
\label{Peregrine2} \\
\ &\times &\left( \frac{3+16iA^{2}t-16A^{4}t^{2}-4A^{2}x^{2}}{%
1+16A^{4}t^{2}+4A^{2}x^{2}}\right) \exp \left( 2iA^{2}t+\sin ^{2}t\right) . 
\notag
\end{eqnarray}%
%
%
%
%
%
%
%
%
%
%
%
%
%
%
%
%
%
%
%
%
%
%
%
%
%
%
%
%
%
%
%
%
%
%
%
%
%
%
%
%
%
\begin{figure}[h]
\centering
\subfigure[Profile of perturbated Peregrine soliton in the times $t=1$ and
$t=5$.]{\includegraphics[scale=0.38]{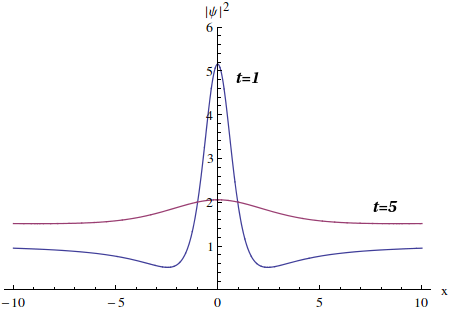}} 
\subfigure[Profile of
perturbated Peregrine soliton at
$x=0$.]{\includegraphics[scale=0.38]{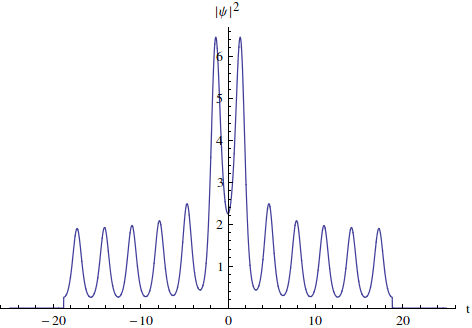}} 
\subfigure[3D view of
perturbated Peregrine soliton.]{\includegraphics[scale=0.36]{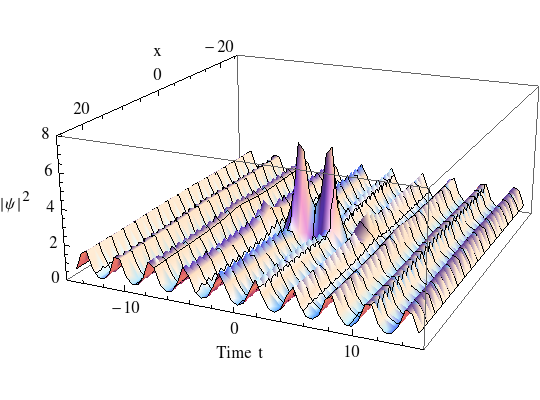}}
\caption{Dynamics of perturbated Peregrine soliton solution (\protect\ref%
{Peregrine2}) with $A=\foreignlanguage{english}{0.5}$ and $\protect\kappa %
(0)=0$ for (\protect\ref{sch2}) with $d(t)=-(\sin (2t)-t)$ and $%
h(t)=-2e^{-2\sin ^{2}(t)}$. }
\label{part2}
\end{figure}

\section{Final Remarks}

In this work, inspired by the work of Mahric on multiparameter solutions for
the linear Schr\"{o}dinger equation with quadratic potential, we have
established a relationship between solutions with parameters of
Riccati-Ermakov systems with the dynamics of nonlinear Schr\"{o}dinger
equations with variable coefficients of the form (\ref{NLSVC}). We have
shown that for special coefficients of (\ref{NLSVC}) it is possible to find
explicit solutions presenting blow up, periodic soliton solutions with
bending properties and more. This work should motivate further analytical
and numerical studies looking to clarify the connections on the dynamics of
variable-coefficient NLS with the dynamics of Riccati-Ermakov systems.

\begin{acknowledgement}
The second named author was supported by the Simons Foundation \#316295. The
authors thanks Drs. A. Mahalov, S. Roudenko and S. K. Suslov for many
different valuable discussions of this project and valuable mentorship.
\end{acknowledgement}

\section{Appendix: Nonlinear Coupled Riccati-Ermakov Systems and Similarity
Transformations}

All the formulas from this appendix have been verified previously in \cite%
{Ko:Su;su}.

\subsection{Modified Riccati system and a similarity transformation}

In section 4 we need a slightly modified nonlinear coupled Riccati system
that includes for convenience of our results a parameter $l_{0}=\pm 1$ (the
case $l_{0}=1$ has already been considered in \cite{Cor-Sot:Lop:Sua:Sus}, 
\cite{Sua1}, \cite{Sus}): \ 
\begin{equation}
\dfrac{d\alpha }{dt}+b(t)+2c(t)\alpha +4a(t)\alpha ^{2}=0,  \label{rica1}
\end{equation}%
\begin{equation}
\dfrac{d\beta }{dt}+(c(t)+4a(t)\alpha (t))\beta =0,  \label{rica2}
\end{equation}%
\begin{equation}
\dfrac{d\gamma }{dt}+l_{0}a(t)\beta ^{2}(t)=0,\quad l_{0}=\pm 1,
\label{rica3}
\end{equation}%
\begin{equation}
\dfrac{d\delta }{dt}+(c(t)+4a(t)\alpha (t))\delta =f(t)+2\alpha (t)g(t),
\label{rica4}
\end{equation}%
\begin{equation}
\dfrac{d\varepsilon }{dt}=(g(t)-2a(t)\delta (t))\beta (t),  \label{rica5}
\end{equation}%
\begin{equation}
\dfrac{d\kappa }{dt}=g(t)\delta (t)-a(t)\delta ^{2}(t).  \label{rica6}
\end{equation}%
\ Considering the standard substitution\ 
\begin{equation}
\alpha =\dfrac{1}{4a(t)}\dfrac{\mu ^{\prime }(t)}{\mu (t)}-\dfrac{d(t)}{2a(t)%
},  \label{sus1}
\end{equation}%
it follows that the Riccati equation (\ref{rica1}) becomes\ 
\begin{equation}
\mu ^{\prime \prime }-\tau (t)\mu ^{\prime }+4\sigma (t)\mu =0,
\label{carac1}
\end{equation}%
with\ 
\begin{equation}
\tau (t)=\frac{a^{\prime }}{a}-2c+4d,\hspace{1cm}\sigma (t)=ab-cd+d^{2}+%
\frac{d}{2}\left( \frac{a^{\prime }}{a}-\frac{d^{\prime }}{d}\right) .
\end{equation}%
\ We will refer to (\ref{carac1}) as the characteristic equation of the
Riccati system. Here $a(t)$, $b(t)$, $c(t)$, $d(t)$, $f(t)$ and $g(t)$ are
real value functions depending only on the variable $t$. A solution of the
Riccati system (\ref{rica1})-(\ref{rica6}) with multiparameters is given by
the following expressions (with the respective inclusion of the parameter $%
l_{0}$) \cite{Cor-Sot:Lop:Sua:Sus}, \cite{Sua1}, \cite{Sus}:\ 
\begin{equation}
\mu \left( t\right) =2\mu \left( 0\right) \mu _{0}\left( t\right) \left(
\alpha \left( 0\right) +\gamma _{0}\left( t\right) \right) ,  \label{mu}
\end{equation}%
\begin{equation}
\alpha \left( t\right) =\alpha _{0}\left( t\right) -\frac{\beta
_{0}^{2}\left( t\right) }{4\left( \alpha \left( 0\right) +\gamma _{0}\left(
t\right) \right) },  \label{alpha}
\end{equation}%
\begin{equation}
\beta \left( t\right) =-\frac{\beta \left( 0\right) \beta _{0}\left(
t\right) }{2\left( \alpha \left( 0\right) +\gamma _{0}\left( t\right)
\right) }=\frac{\beta \left( 0\right) \mu \left( 0\right) }{\mu \left(
t\right) }w\left( t\right) ,  \label{beta}
\end{equation}%
\begin{equation}
\gamma \left( t\right) =l_{0}\gamma \left( 0\right) -\frac{l_{0}\beta
^{2}\left( 0\right) }{4\left( \alpha \left( 0\right) +\gamma _{0}\left(
t\right) \right) },\quad l_{0}=\pm 1,  \label{gamma}
\end{equation}%
\begin{equation}
\delta \left( t\right) =\delta _{0}\left( t\right) -\frac{\beta _{0}\left(
t\right) \left( \delta \left( 0\right) +\varepsilon _{0}\left( t\right)
\right) }{2\left( \alpha \left( 0\right) +\gamma _{0}\left( t\right) \right) 
},  \label{delta}
\end{equation}%
\begin{equation}
\varepsilon \left( t\right) =\varepsilon \left( 0\right) -\frac{\beta \left(
0\right) \left( \delta \left( 0\right) +\varepsilon _{0}\left( t\right)
\right) }{2\left( \alpha \left( 0\right) +\gamma _{0}\left( t\right) \right) 
},  \label{epsilon}
\end{equation}%
\begin{equation}
\kappa \left( t\right) =\kappa \left( 0\right) +\kappa _{0}\left( t\right) -%
\frac{\left( \delta \left( 0\right) +\varepsilon _{0}\left( t\right) \right)
^{2}}{4\left( \alpha \left( 0\right) +\gamma _{0}\left( t\right) \right) },
\label{kappa}
\end{equation}%
\ subject to the initial arbitrary conditions $\mu \left( 0\right) ,$ $%
\alpha \left( 0\right) ,$ $\beta \left( 0\right) \neq 0,$ $\gamma (0),$ $%
\delta (0),$ $\varepsilon (0)$ and $\kappa (0)$. $\alpha _{0}$, $\beta _{0}$%
, $\gamma _{0}$, $\delta _{0}$, $\varepsilon _{0}$ and $\kappa _{0}$ are
given explicitly by\ 
\begin{equation}
\alpha _{0}\left( t\right) =\frac{1}{4a\left( t\right) }\frac{\mu
_{0}^{\prime }\left( t\right) }{\mu _{0}\left( t\right) }-\frac{d\left(
t\right) }{2a\left( t\right) },  \label{alpha0}
\end{equation}%
\begin{equation}
\beta _{0}\left( t\right) =-\frac{w\left( t\right) }{\mu _{0}\left( t\right) 
},\quad w\left( t\right) =\exp \left( -\int_{0}^{t}\left( c\left( s\right)
-2d\left( s\right) \right) \ ds\right) ,  \label{beta0}
\end{equation}%
\begin{equation}
\gamma _{0}\left( t\right) =\frac{d\left( 0\right) }{2a\left( 0\right) }+%
\frac{1}{2\mu _{1}\left( 0\right) }\frac{\mu _{1}\left( t\right) }{\mu
_{0}\left( t\right) },  \label{gamma0}
\end{equation}%
\begin{equation}
\delta _{0}\left( t\right) =\frac{w\left( t\right) }{\mu _{0}\left( t\right) 
}\ \ \int_{0}^{t}\left[ \left( f\left( s\right) -\frac{d\left( s\right) }{%
a\left( s\right) }g\left( s\right) \right) \mu _{0}\left( s\right) +\frac{%
g\left( s\right) }{2a\left( s\right) }\mu _{0}^{\prime }\left( s\right) %
\right] \ \frac{ds}{w\left( s\right) },  \label{delta0}
\end{equation}%
\begin{eqnarray}
\varepsilon _{0}\left( t\right) &=&-\frac{2a\left( t\right) w\left( t\right) 
}{\mu _{0}^{\prime }\left( t\right) }\delta _{0}\left( t\right)
+8\int_{0}^{t}\frac{a\left( s\right) \sigma \left( s\right) w\left( s\right) 
}{\left( \mu _{0}^{\prime }\left( s\right) \right) ^{2}}\left( \mu
_{0}\left( s\right) \delta _{0}\left( s\right) \right) \ ds  \label{epsilon0}
\\
&&+2\int_{0}^{t}\frac{a\left( s\right) w\left( s\right) }{\mu _{0}^{\prime
}\left( s\right) }\left[ f\left( s\right) -\frac{d\left( s\right) }{a\left(
s\right) }g\left( s\right) \right] \ ds,  \notag
\end{eqnarray}%
\begin{eqnarray}
\kappa _{0}\left( t\right) &=&\frac{a\left( t\right) \mu _{0}\left( t\right) 
}{\mu _{0}^{\prime }\left( t\right) }\delta _{0}^{2}\left( t\right)
-4\int_{0}^{t}\frac{a\left( s\right) \sigma \left( s\right) }{\left( \mu
_{0}^{\prime }\left( s\right) \right) ^{2}}\left( \mu _{0}\left( s\right)
\delta _{0}\left( s\right) \right) ^{2}\ ds  \label{kappa0} \\
&&\quad -2\int_{0}^{t}\frac{a\left( s\right) }{\mu _{0}^{\prime }\left(
s\right) }\left( \mu _{0}\left( s\right) \delta _{0}\left( s\right) \right) %
\left[ f\left( s\right) -\frac{d\left( s\right) }{a\left( s\right) }g\left(
s\right) \right] \ ds,  \notag
\end{eqnarray}%
\ with $\delta _{0}\left( 0\right) =g_{0}\left( 0\right) /\left( 2a\left(
0\right) \right) ,$ $\varepsilon _{0}\left( 0\right) =-\delta _{0}\left(
0\right) ,$ $\kappa _{0}\left( 0\right) =0.$ Here $\mu _{0}$ and $\mu _{1}$
represent the fundamental solution of the characteristic equation subject to
the initial conditions $\mu _{0}(0)=0$, $\mu _{0}^{\prime }(0)=2a(0)\neq 0$
and $\mu _{1}(0)\neq 0$, $\mu _{1}^{\prime }(0)=0$.

Using the system (\ref{alpha})-(\ref{kappa}), in \cite{Sus} we see a
generalized lens transformation is presented. Next we recall this result
(here we present a slight perturbation introducing the parameter $l_{0}=\pm
1 $ in order to use Peregrine-type soliton solutions):

\begin{lemma}[$l_{0}=1$, \protect\cite{Sus}]
Assume that $h(t)=\lambda a(t)\beta ^{2}(t)\mu (t)$ with $\lambda \in 
\mathbb{R}$. Then the substitution\ 
\begin{equation}
\psi (t,x)=\dfrac{1}{\sqrt{\mu (t)}}e^{i(\alpha (t)x^{2}+\delta (t)x+\kappa
(t))}u(\tau ,\xi ),  \label{tra1}
\end{equation}%
\ where $\xi =\beta \left( t\right) x+\varepsilon \left( t\right) $ and $%
\tau =\gamma \left( t\right) $, transforms the equation \ 
\begin{equation*}
i\psi _{t}=-a(t)\psi _{xx}+b(t)x^{2}\psi -ic(t)x\psi _{x}-id(t)\psi
-f(t)x\psi +ig(t)\psi _{x}+h(t)|\psi |^{2}\psi
\end{equation*}%
\ into the standard Schr$\ddot{\mbox{o}}$dinger equation 
\begin{equation}
\\
\ iu_{\tau }-l_{0}u_{\xi \xi }+l_{0}\lambda |u|^{2}u=0,\quad \quad l_{0}=\pm
1,  \label{estandar}
\end{equation}%
\ as long as $\alpha ,$ $\beta ,$ $\gamma ,$ $\delta ,$ $\varepsilon $ and $%
\kappa $ satisfy the Riccati system (\ref{rica1})-(\ref{rica6}) and also
equation (\ref{sus1}).
\end{lemma}

\subsection{Ermakov System and a Similarity Transformation}

We recall the following useful results for sections 2 and 3.

\begin{lemma}[\protect\cite{Lan:Lop:Sus}, \protect\cite{Lo:Su:VeSy} and 
\protect\cite{Lop:Sus:VegaGroup}]
The following nonlinear coupled system (Ermakov system) 
\begin{align}
& \frac{d\alpha }{dt}+b+2c\alpha +4a\alpha ^{2}=c_{0}a\beta ^{4},
\label{Erma11} \\
& \frac{d\beta }{dt}+\left( c+4a\alpha \right) \beta =0,  \label{Erma12} \\
& \frac{d\gamma }{dt}+a\beta ^{2}=0,  \label{Erma13} \\
& \frac{d\delta }{dt}+\left( c+4a\alpha \right) \delta =f+2cg+2c_{0}a\beta
^{3}\varepsilon ,  \label{Erma14} \\
& \frac{d\varepsilon }{dt}=\left( g-2a\delta \right) \beta ,  \label{Erma15}
\\
& \frac{d\kappa }{dt}=g\delta -a\delta ^{2}+c_{0}a\beta ^{2}\varepsilon ^{2}
\label{Erma16} \\
& \alpha \left( t\right) =-\frac{1}{4a\left( t\right) }\frac{\mu ^{\prime
}\left( t\right) }{\mu \left( t\right) }-\frac{d\left( t\right) }{2a\left(
t\right) }  \label{Erma17}
\end{align}

admits the following multiparameter solution given explicitly by%
\begin{align}
& \mu \left( t\right) =\mu _{0}(t)\mu (0)\sqrt{4\left( \gamma _{0}(t)+\alpha
(0)\right) ^{2}+\beta ^{4}(0)},  \label{SErma11} \\
& \alpha \left( t\right) =\alpha _{0}(t)-\frac{\beta _{0}^{2}(t)\left(
\gamma _{0}(t)+\alpha (0)\right) }{4\left( \gamma _{0}(t)+\alpha (0)\right)
^{2}+\beta ^{4}(0)},  \label{SErma12} \\
& \beta \left( t\right) =-\frac{\beta \left( 0\right) \beta _{0}(t)}{\sqrt{%
4\left( \gamma _{0}(t)+\alpha \left( 0\right) \right) ^{2}-\beta ^{4}(0)}},
\label{SErma13} \\
& \gamma \left( t\right) =\gamma \left( 0\right) -\frac{1}{2}\arctan \frac{%
\beta ^{2}(0)}{2(\gamma _{0}(t)+\alpha (0))}  \label{SErma14}
\end{align}%
and 
\begin{align}
& \delta \left( t\right) =\delta _{0}(t)-\beta _{0}(t)\frac{\varepsilon
(0)\beta ^{3}(0)+2(\gamma _{0}(t)+\alpha (0))(\varepsilon _{0}(t)+\delta (0))%
}{4(\gamma _{0}(t)+\alpha (0))^{2}+\beta ^{4}(0)}  \label{SErma16} \\
& \varepsilon \left( t\right) =\frac{-\beta (0)\left( \delta (0)+\varepsilon
_{0}(t)\right) +2\varepsilon (0)\left( \gamma _{0}(t)+\alpha (0)\right) }{%
\sqrt{4\left( \gamma _{0}(t)+\alpha (0)\right) ^{2}+\beta ^{4}(0)}},
\label{SErma17} \\
& \kappa \left( t\right) =\kappa _{0}(t)+\kappa (0)-\frac{\beta
^{3}(0)\varepsilon (0)(\varepsilon _{0}(t)+\delta (0))}{4(\gamma
_{0}(t)+\alpha (0))^{2}+\beta ^{4}(0)}  \label{SErma18} \\
& \qquad \quad \quad +\frac{(\gamma _{0}(t)+\alpha (0))\left[ \beta
^{2}(0)\varepsilon ^{2}(0)-(\varepsilon _{0}(t)+\delta (0))^{2}\right] }{%
4(\gamma _{0}(t)+\alpha (0))^{2}+\beta ^{4}(0)},  \notag
\end{align}%
subject to arbitrary initial data $\mu \left( 0\right) ,$ $\alpha \left(
0\right) ,$ $\beta \left( 0\right) \neq 0,$ $\gamma \left( 0\right) ,$ $%
\delta \left( 0\right) ,$ $\varepsilon \left( 0\right) ,$ $\kappa \left(
0\right) $ where $\alpha _{0}(t),$ $\beta _{0}(t),$ $\gamma _{0}(t),$ $%
\delta _{0}(t),$ $\epsilon _{0}(t)$ and $\kappa _{0}(t)$ are given by (\ref%
{alpha0})-(\ref{kappa0}).
\end{lemma}

We will also need a 2D version of the results above for the blow-up results
of Section 2:

\begin{lemma}
(\cite{Ma:Suslov})The nonlinear equation%
\begin{align}
i\psi _{t}& =-a\left( \psi _{xx}+\psi _{yy}\right) +b\left(
x^{2}+y^{2}\right) \psi -ic\left( x\psi _{x}+y\psi _{y}\right) -2id\psi
\label{NonlinearParabolic2D} \\
& -\left( xf_{1}+yf_{2}\right) \psi +i\left( g_{1}\psi _{x}+g_{2}\psi
_{y}\right) +h\left\vert \psi \right\vert ^{2s}\psi ,  \notag
\end{align}%
where $a,$ $b,$ $c,$ $d,$ $f_{1,2}$ and $g_{1,2}$ are real-valued functions
of $t,$ can be transformed to%
\begin{equation}
i\chi _{\tau }-l_{0}(\chi _{\xi \xi }+\chi _{\eta \eta
})=-l_{0}h_{0}\left\vert \chi \right\vert ^{2s}\chi \qquad \left( l_{0}=\pm
1\right)  \label{2DShroedingerNLTransformed}
\end{equation}%
by the ansatz 
\begin{equation}
\psi =\mu ^{-1}e^{i(\alpha (x^{2}+y^{2})+(\delta _{1}x+\delta _{2}y)+\kappa
_{1}+\kappa _{2})}\chi (\xi ,\eta ,\tau ),  \label{2d transformation}
\end{equation}

where $\xi =\beta (t)x+\varepsilon _{1}(t),$ $\eta =\beta (t)y+\varepsilon
_{2}(t),$ $\tau =\gamma (t),$ $h(t)=h_{0}a(t)\beta ^{2}(t)\mu ^{2s}(t)$\ $%
(h_{0}$ is a constant), provided that%
\begin{align}
\frac{d\alpha }{dt}+b+2c\alpha +4a\alpha ^{2}& =0,  \label{SysA} \\
\frac{d\beta }{dt}+(c+4a\alpha )\beta & =0,  \label{SysB} \\
\frac{d\gamma }{dt}+a\beta ^{2}& =0,  \label{SysC} \\
\frac{d\delta _{1,2}}{dt}+(c+4a\alpha )\delta _{1,2}& =f_{1,2}+2g\alpha ,
\label{SysD} \\
\frac{d\varepsilon _{1,2}}{dt}& =(g-2a\delta _{1,2})\beta ,  \label{SysE} \\
\frac{d\kappa _{1,2}}{dt}& =g\delta _{1,2}-a\delta _{1,2}^{2}.  \label{SysF}
\end{align}%
Here,%
\begin{equation}
\alpha =\frac{1}{4a}\frac{\mu ^{\prime }}{\mu }-\frac{d}{2a},  \label{Alpha}
\end{equation}%
and solutions of the system (\ref{SysA})--(\ref{SysF}) are given by (\ref{mu}%
)--(\ref{kappa}).
\end{lemma}

\end{document}